\makeatletter
\let\latexkernel@label\label
\makeatother

\documentclass[aps,prl,superscriptaddress,twocolumn,nofootinbib,notitlepage,10pt]{revtex4-2}

\makeatletter
\let\label\latexkernel@label
\makeatother
\usepackage{amsthm}
\usepackage{amsmath,amssymb,mathtools}
\usepackage{physics}
\usepackage{bbm}
\usepackage{color}
\usepackage{booktabs}
\usepackage{tikz}
\usepackage{cancel}
\usepackage{graphicx}
\usepackage{adjustbox}
\usetikzlibrary{arrows.meta,positioning,calc,fit,decorations.pathreplacing}

\newtheorem{theorem}{Theorem}[section]
\newtheorem{lemma}{Lemma}[section]
\newtheorem{corollary}{Corollary}[section]

\theoremstyle{definition}
\newtheorem{definition}{Definition}[section]
\newtheorem{problem}{Problem}[section]

\renewcommand{\thetheorem}{\ifnum\value{section}=0
    \arabic{theorem}\else
    \thesection.\arabic{theorem}\fi
}
\renewcommand{\theproblem}{\ifnum\value{section}=0
    \arabic{problem}\else
    \thesection.\arabic{problem}\fi
}
\renewcommand{\thecorollary}{\ifnum\value{section}=0
    \arabic{corollary}\else
    \thesection.\arabic{corollary}\fi
}

\newcommand{\Id}{\mathbbm{1}}
\newcommand{\pauli}{\mathbb P}
\newcommand{\STAB}{\mathrm{STAB}}
\newcommand{\stp}{\mathrm{SP}}
\newcommand{\sharpP}{\#\mathrm P}
\newcommand{\GapP}{\mathrm{GapP}}
\newcommand{\CeqP}{\mathrm C_{=}\mathrm P}
\newcommand{\PP}{\mathrm{PP}}
\newcommand{\FP}{\mathrm{FP}}
\newcommand{\poly}{\mathrm{poly}}
\newcommand{\cont}{\mathcal C}
\newcommand{\cc}[1]{\mathrm{#1}}
\newcommand{\NP}{\cc{NP}}
\newcommand{\coNP}{\cc{coNP}}
\newcommand{\BQP}{\cc{BQP}}
\newcommand{\dist}{\operatorname{dist}}
\newcommand{\ApproxSTABPEPS}{\mathrm{Approx\text{-}STAB\text{-}PEPS}}
\newcommand{\stabpeps}{\mathrm{STAB\text{-}PEPS}_{\mathrm{exact}}}
\newcommand{\yes}{\mathrm{YES}}
\newcommand{\no}{\mathrm{NO}}
\newcommand{\ceil}[1]{\left\lceil #1 \right\rceil}
\newcommand{\floor}[1]{\left\lfloor #1 \right\rfloor}
\newcommand{\ot}{\otimes}
\newcommand{\unpeps}{\ket*{\widetilde{\psi}(A)}}
\newcommand{\lstab}{\mathrm{loc STAB}_{AB}}
\definecolor{myrefcolor}{rgb}{0.067,0.5,0.5}
\usepackage[colorlinks=true,citecolor=myrefcolor,urlcolor=myrefcolor,linkcolor=myrefcolor]{hyperref}
\usepackage[capitalize]{cleveref}
 \usepackage{verbatim}

\begin{document}

\title{Nonstabilizerness of quantum tensor network states is intractable in two dimensions}

\author{Gianluca Esposito}
\affiliation{Scuola Superiore Meridionale, Largo S. Marcellino 10, 80138 Naples, Italy}
\affiliation{Istituto Nazionale di Fisica Nucleare, Sezione di Napoli}
\affiliation{Max Planck Institute of Quantum Optics, Hans Kopfermann Str. 1, 85748 Garching, Germany}
\affiliation{Munich Center for Quantum Science and Technology, Schellingstr. 4, 80799 Munich, Germany}
\author{Georgios Styliaris}
\affiliation{Max Planck Institute of Quantum Optics, Hans Kopfermann Str. 1, 85748 Garching, Germany}
\affiliation{Munich Center for Quantum Science and Technology, Schellingstr. 4, 80799 Munich, Germany}
\date{\today}

\begin{abstract}
Nonstabilizerness is a necessary resource for quantum systems to lie beyond the classically simulable regime. With the advent of stabilizer $\alpha$-Rényi entropies, nonstabilizerness has also become a many-body diagnostic, complementary to entanglement. While deciding whether an arbitrary quantum state has nonstabilizerness is provably hard, such states already require a description exponentially large in the number of qubits and are thus out of reach for many-body physics. Here we instead consider 2D tensor network (TN) states, which compactly capture states obeying an entanglement area law, and ask whether they admit a simpler algorithm for the same task. In contrast to the 1D case, we prove that, even at small, constant bond dimension, computing the stabilizer entropy of 2D TN states is $\#\mathrm{P}$-hard for any integer $\alpha\ge2$, and deciding stabilizer membership is $\mathrm{C_= P}$-complete. The corresponding constant-accuracy problems remain $\mathrm{C_= P}$-hard. We further prove that deciding whether a PEPS can be transformed into a stabilizer state by local unitaries over a specified bipartition is $\mathrm{C_= P}$-hard. Under standard complexity assumptions, these results rule out classical or quantum algorithms with polynomial resources for all three tasks in the worst case.
\end{abstract}

\maketitle

\textit{Introduction}---Entanglement alone does not separate which quantum states and processes are classically tractable from those which are not. Stabilizer states can be highly entangled, yet they are efficiently simulable on classical computers~\cite{Gottesman_1998}. A necessary missing ingredient is nonstabilizerness, the resource that is required to achieve universal quantum computation~\cite{Veitch_Hamed_Mousavian_Gottesman_Emerson_2014,bravyi_universal,bravyi2016improved}. Beyond this role, nonstabilizerness has become a many-body diagnostic, marking phase transitions~\cite{Viscardi_2026,adhikary2026quantumcomplexitythermalphase,Catalano_Odavić_Torre_Hamma_Franchini_Giampaolo_2024,PRXQuantum.3.020333,Sarkar2020,Tarabunga_Tirrito_Chanda_Dalmonte_2023} and the onset of chaos in quantum circuits~\cite{Leone2021quantumchaosis,Turkeshi_Tirrito_Sierant_2024} and unitary dynamics~\cite{Cusumano_Esposito_Hamma_2026,iannotti2026nonstabilizernessu1symmetrychaotic,tirrito2025magicphasetransitionsmonitored}. Finer, basis-independent notions of non-local and long-range nonstabilizerness~\cite{cao2025gravitational,nlmagic,frustration_nl,iannotti2026nonlocalmagicresourcesfermionic,collura2026nonlocalnonstabilizernessfreefermion,malvimat2026multipartitenonlocalmagicsyk,pyzr-jmvw,Korbany_2025,korbany2026longrangenonstabilizernesstopologicallyencoded,wei2026longrangenonstabilizernessquantumcodes,nehra2025topologicalmagicresponsequantum,viscardi2026nonlocalmagicclosedformsolution,sierant2026exactquantificationnonlocalmagic} have been introduced, which can be used to probe topological order. One ingredient that makes such a program viable is the availability of nonstabilizerness measures that can be directly evaluated. Stabilizer entropies (SE)~\cite{Leone_Oliviero_Hamma_2022} have emerged as a useful set of measures~\cite{Bittel2026operational,esposito2026stabilizerentropytrustworthymixed}, as they avoid the optimization over stabilizer decompositions characteristic of other monotones~\cite{Leone_Bittel_2024,Howard_Campbell_2017}.

\begin{figure}[t!]
    \centering
    \includegraphics[width=\linewidth]{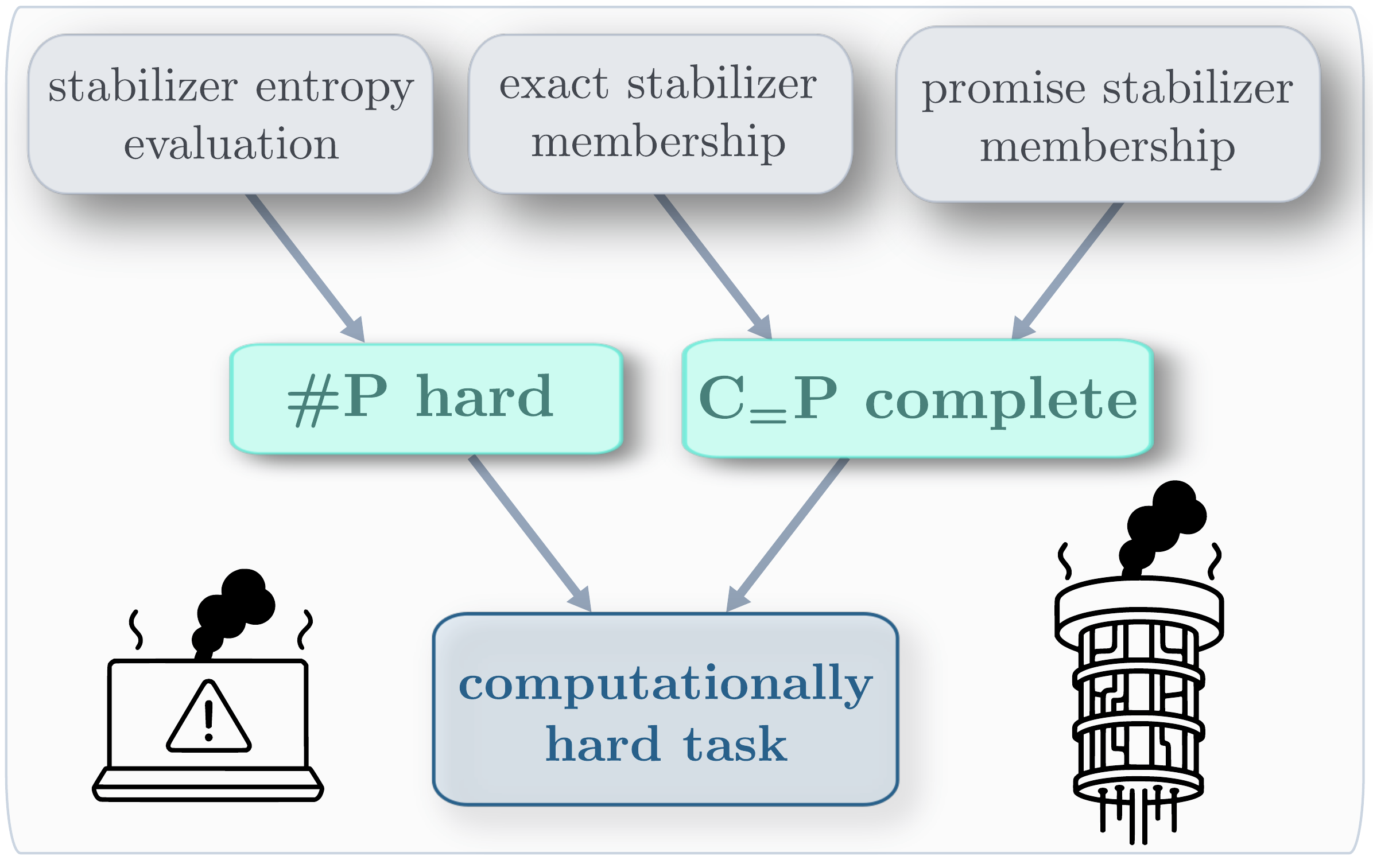}
    \caption{We establish that computing SE of 2D PEPS or just deciding whether a PEPS is a stabilizer state, exactly or approximately, are hard computational tasks for classical and quantum computers. An efficient classical algorithm for either of those tasks would imply the collapse of the polynomial hierarchy, and an efficient quantum algorithm would imply the collapse $\NP\subseteq \BQP$: both implications are considered highly unlikely.
    }
    \label{fig:results}
\end{figure}

This raises a computational question: given a description of a state, how hard is it to quantify its nonstabilizerness, or even to decide whether it has any at all? An exact algorithm computes SEs at an exponential cost in the number of qubits~\cite{huang2026fastexactapproachstabilizer}, and it has been proven that for arbitrary mixed states the problem is intractable: deciding membership in the stabilizer set is super-exponentially hard even in an approximate sense~\cite{unbearable}. This formulation, however, takes the full density matrix as input, an object whose size grows exponentially with the number of qubits and is already out of practical reach after a few tens of qubits. This problem has also been tackled in a sparse and a marginal-probability setting~\cite{Varela_Keller_Junior_Moreira_Chaves_Macêdo_2026}, but this also does not apply to many-body systems, which in general have exponentially many non-zero Pauli weights.
Generic quantum states are, in that sense, too general to be of relevance to many-body physics, which deals with systems having significantly more structure in their correlations. Tensor networks (TN) provide a canonical family making these correlations explicit, while also providing a compact description. Their central idea is to replace the exponentially large number of amplitudes of a many-body wavefunction by a contraction network of tensors, yielding a description of size polynomial in the number of qubits whenever the bond dimension grows slow enough. The resulting class is that of projected entangled pair states (PEPS)~\cite{Cirac_Perez-Garcia_Schuch_Verstraete_2021}, and their one-dimensional counterpart, matrix product states (MPS)~\cite{Perez-Garcia_Verstraete_Wolf_Cirac_2007,fannes1992finitely}. These families capture states relevant for many-body physics, as they efficiently encode the entanglement area law~\cite{Eisert2010}, obeyed by ground states of gapped, local Hamiltonians~\cite{Hastings_2006,arad2013area,anshu2022area}, and in two dimensions represent exactly, at constant bond dimension, all known (non-chiral) topological phases~\cite{levin2005string,buerschaper2009explicit,gu2009tensor}.

It is then natural to ask how hard it is to evaluate stabilizer R\'enyi entropies when the input is a PEPS of constant bond dimension, specified efficiently by its tensors, rather than a generic state. Note that the hardness results for unstructured inputs do not answer this question, as the hard instances they produce need not admit a tensor-network representation at bounded bond dimension. At the same time, the ansatz has additional structure that an algorithm could exploit, so tractability is not excluded either. The answer has thus to be established for the tensor-network input model itself, and it need not be the same one as for generic states. Indeed, in one spatial dimension the MPS restriction drastically changes the complexity: stabilizer R\'enyi entropies of MPS can be computed at a cost polynomial in the system size and in the bond dimension~\cite{Haug_Piroli_2023,PhysRevLett.131.180401,PhysRevLett.133.010602}. This efficiency is however restricted to one dimension, and thus excludes 2D locally interacting systems, such as those with topological order.

In this Letter, we answer this question by determining the worst-case complexity of SE evaluation and stabilizer membership decision for 2D PEPS (\cref{fig:results}). In short, we show that neither of these problems can be solved with polynomial resources, on a classical or a quantum computer, even for constant bond dimension to an additive precision. This is subject to the standard complexity assumption that quantum computers cannot efficiently solve all $\mathrm {NP}$ problems~\cite{BennettBernsteinBrassardVazirani1997Strengths,Aaronson2005NPComplete}. 

In particular, we first show that SE evaluation is at least as hard as $\sharpP$-counting~\cite{valiant_permanent}, already for 2D qubit PEPS of bond dimension $D=4$. This renders SE evaluation as hard as arbitrary PEPS contractions~\cite{Schuch_Wolf_Verstraete_Cirac_2007}, including exactly counting the solutions of $\mathrm{NP}$ problems. Moreover, we prove that computational hardness is retained when computing SE in PEPS up to constant accuracy. 

Often, however, the goal is not to quantify nonstabilizerness, but only to decide whether a state is a stabilizer: this weaker decision problem is known as stabilizer membership. We study its complexity, both exactly and up to constant accuracy, for small bond dimension 2D PEPS, and show that both versions are at least as hard as the complexity class $\CeqP$~\cite{gapP}, which amounts to deciding whether a Boolean expression admits as many YES as NO instances. An efficient classical or quantum algorithm for either is highly unlikely, as it would imply a collapse of the classical polynomial hierarchy or the inclusion $\rm NP \subseteq \BQP$, respectively. Certifying the stabilizer structure from a 2D PEPS description thus remains intractable in the worst case, even at small bond dimension.

We further prove that hardness extends to non-local nonstabilizerness: deciding whether a given PEPS can be transformed into a stabilizer state by using factorized unitaries across a bipartition is $\CeqP$-hard, already for square-lattice qubit PEPS of constant bond dimension.

\textit{Stabilizer formalism}---Let $\pauli_N:=\{\Id,X,Y,Z\}^{\otimes N}$ be the set of phase-free Pauli strings on $N$ qubits, and let $\STAB_N$ be the set of pure $N-$qubit stabilizer states, namely $+1$ eigenstates of $N$ algebraically independent commuting Pauli strings \cite{nielsen_chuang}. Stabilizer states, and the dynamics that preserve them (namely the Clifford group), can be simulated efficiently on a classical computer~\cite{nielsen_chuang,Gottesman_1998}.
In turn, when a state evolves under non-Clifford dynamics, this efficient simulation scheme gradually breaks down~\cite{bravyi2016improved}: the property of being outside of the $\STAB_N$ set is called nonstabilizerness. In this work, we use stabilizer entropy (SE)~\cite{Leone_Oliviero_Hamma_2022} to quantify nonstabilizerness. Given a pure state $\ket{\psi}$, we define its characteristic distribution over $\pauli_N $ as $\Xi_P(\psi):=2^{-N}\Tr(P\ketbra{\psi})^2$. This distribution quantifies the probability of finding the Pauli string $P$ in the state representation of $\ket{\psi}$. Then, the $\alpha-$SE, for an integer $\alpha>1$, is defined as the (offset) $\alpha-$R\'enyi entropy of the characteristic distribution:
\begin{equation}
\label{main:se_def}
        M_\alpha (\psi):=\frac{1}{1-\alpha}\log \left(2^{N(\alpha-1)}\sum_{P}\Xi_P(\psi)^\alpha\right)\,,
\end{equation}
with the logarithm to be intended in base two.
Pure normalized stabilizer states satisfy the condition $M_\alpha(\psi)=0\iff \ket\psi\in\STAB_N$ for all $\alpha$, with higher values attained by non-stabilizer states~\cite{Leone_Oliviero_Hamma_2022,Bittel2026operational}: SE quantifies the flatness of the characteristic distribution, and stabilizer states minimize this spread.

\textit{Projected entangled pair states}---We now introduce PEPS for qubits on a square lattice; a thorough introduction can be found in Ref.~\cite{Cirac_Perez-Garcia_Schuch_Verstraete_2021}. Given such a lattice $\Lambda=(V,E)$, we first place a $D-$dimensional maximally entangled state $\ket{\Omega}_e$ over each edge $e \in E$ and then apply a local linear map $A_v: (\mathbb{C}^D)^{\otimes \deg (v)} \to \mathbb{C}^2$ at each vertex $v \in V$, where $\deg (v)$ denotes the degree (e.g., 4 in the bulk). The resulting (unnormalized) state
\begin{equation}
        \unpeps=\left(\bigotimes_{v\in V}A_v\right)\bigotimes_{e\in E }\ket{\Omega}_e\,=\!\!\!\!\raisebox{-0.35\height}{\includegraphics[height=3.5 em]{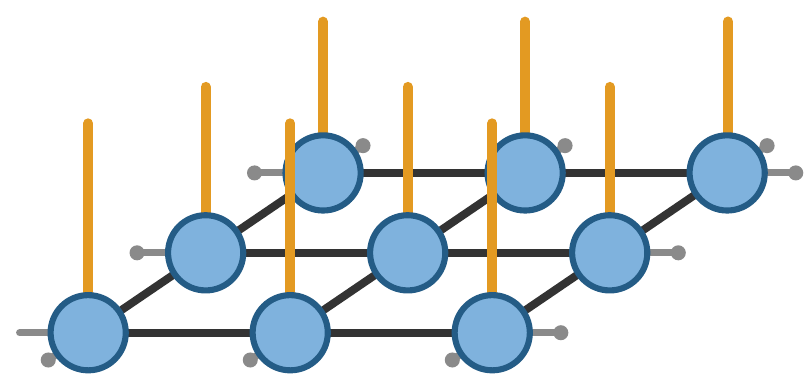}}
        \label{main:eq_pepsdef}
\end{equation}
is a PEPS of \emph{bond dimension} $D$, specified by local maps $\{ A_v \}$, and represents a state of qubits placed on the vertices of $\Lambda$. Introducing a computational basis, each $A_v$ can be represented by a \emph{tensor}, illustrated graphically as 
\begin{equation}
    \mel{i}{A_v}{j_1 j_2 j_3 j_4}=\raisebox{-0.4\height}{\includegraphics[height=4 em]{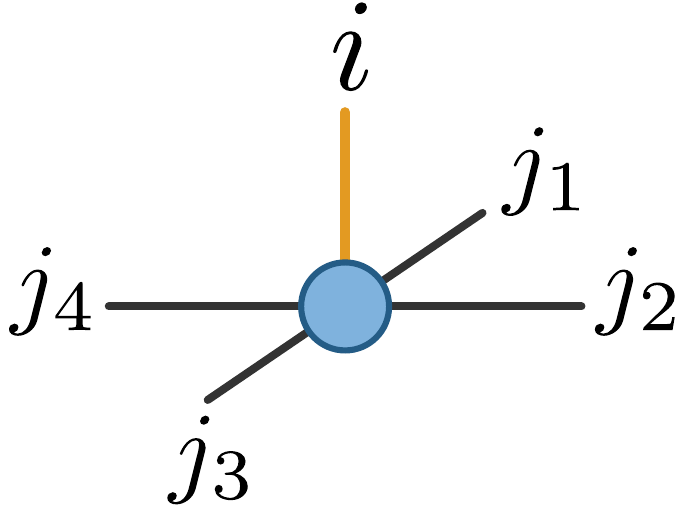}}\in \mathbb{C}\,.
\end{equation}
We refer to the $j$ tensor legs as \emph{virtual}.

State-specific quantities of $\unpeps$, like computational basis amplitudes, its norm, and expectation values of local observables, are obtained by tensor contractions, denoted as $\mathcal{C}(A)$, which consist of summing over virtual tensor indices: for instance, the evaluation of its norm reads
\begin{equation}
    \braket*{\widetilde{\psi}(A)}=\raisebox{-0.33\height}{\includegraphics[height=5 em]{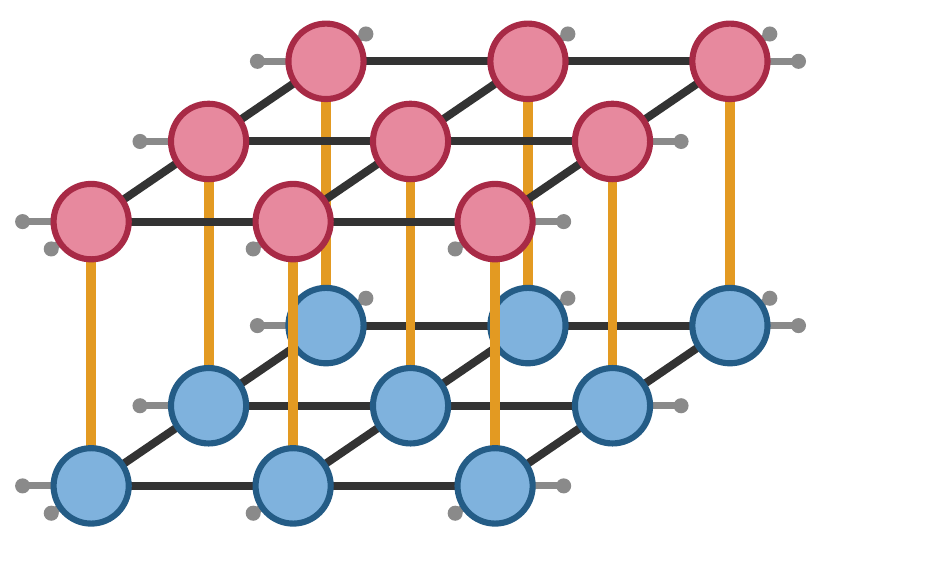}}\,,
\end{equation}
where the red circles represent tensors with complex-conjugated entries.
Despite their simple representations, PEPS contractions in 2D are known to be computationally hard to carry out already for $D = 2$~\cite{Schuch_Wolf_Verstraete_Cirac_2007}: contracting generic PEPS would allow solving $\sharpP$ problems (namely, counting the satisfying assignments of a given Boolean formula~\cite{valiant_permanent,AroraBarak2009}): this class contains $\mathrm{NP}$ and is not expected to be tractable with polynomial resources for either classical or quantum computers. This hardness result has also been refined, proving that $\sharpP$ hardness persists for the average case~\cite{Haferkamp_Hangleiter_Eisert_Gluza_2020}, and computing local expectation values in more restricted tensor-network families, such as isometric TN states~\cite{Malz_Trivedi_2024}, exhibits $\BQP-$hardness.

\begin{figure}[t!]
        \includegraphics[width=\linewidth]{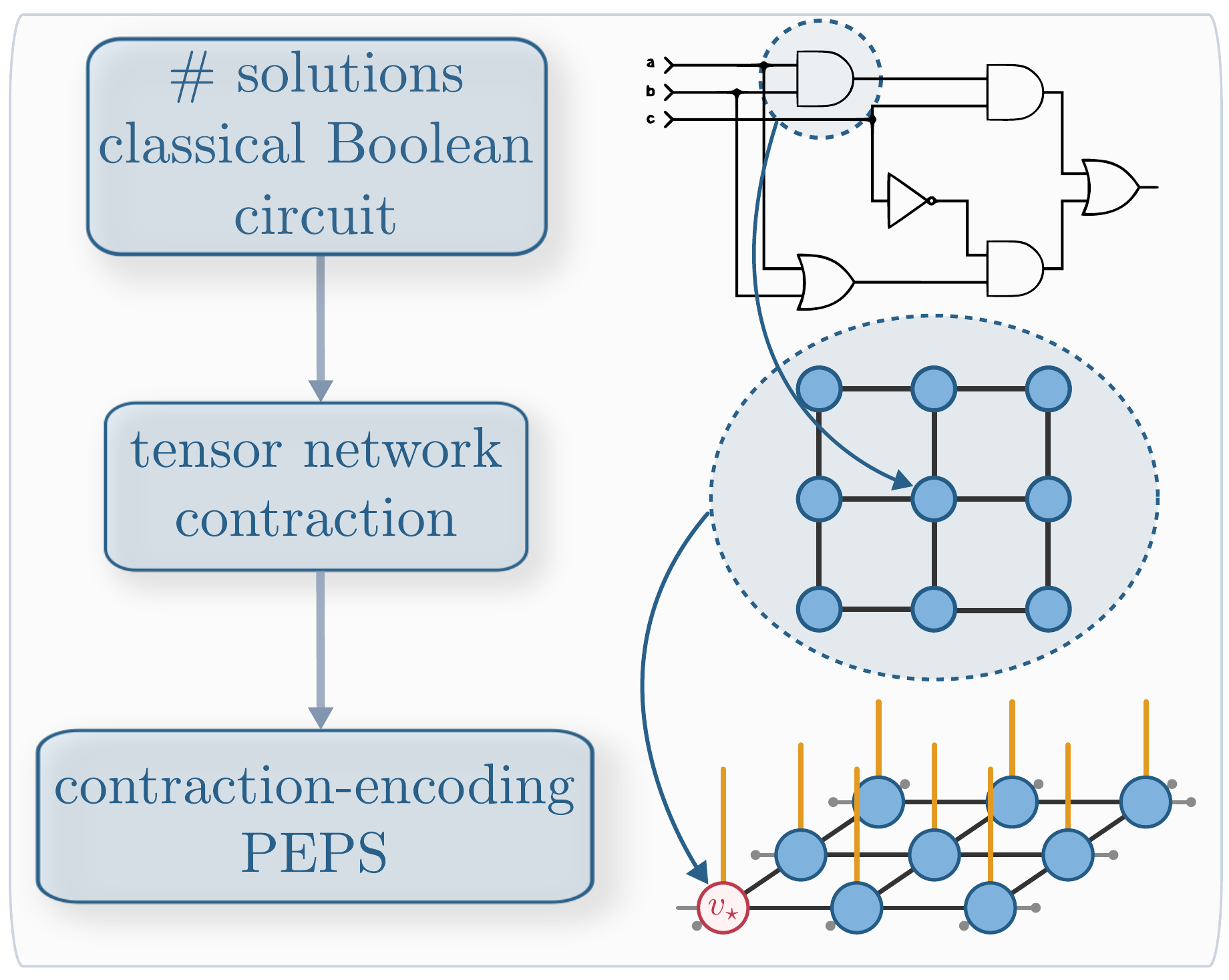}
        \caption{Graphical schematizations of the proofs of computational hardness in this work. The main idea is always to encode the number of satisfying assignments of a Boolean circuit in a scalar TN contraction, then embed its result in the 1-qubit amplitude of a specific PEPS.}
        \label{fig:methods}
\end{figure}

\textit{Stabilizer entropy evaluation}---How hard is it to compute the amount of nonstabilizerness in PEPS? Focusing on the SE, we have:  
\begin{problem}[Exact PEPS SE evaluation]
\label{main:prob_sp}
 Given a rational tensor description $\{A_i\}$ of a PEPS on a lattice with bond dimension $D$, output a dyadic rational number (namely, a rational with its denominator a power of two) that approximates the $\alpha-$SE of the corresponding normalized PEPS $\ket{\psi(A)}$ up to $m$ digits.
\end{problem}

Focusing on the $\alpha=2$ case for simplicity (but without loss of generality), we now show that computing SE of a PEPS can be recast as two PEPS contractions. In particular, the argument of the logarithm in Eq.~\eqref{main:se_def} can be rewritten using the replica trick:
\begin{equation}
    M_2(\ket*{\psi(A)})=-\log \left[2^N\frac{\bra*{\widetilde{\psi}(A)}^{\otimes 4}Q_N \ket*{\widetilde{\psi}(A)}^{\otimes 4}}{\braket*{\widetilde{\psi}(A)}^4}\right]
\end{equation}
with $Q_N:=4^{-N}\sum_{P\in\pauli_N} P^{\otimes 4}$ being an orthogonal projector acting on four copies of $(\mathbb{C}^{2})^{\otimes N}$ and that satisfies $Q_N=\bigotimes_{i=1}^N Q_1$. Applying this formula to a PEPS with tensors $A = \{A_i\}$ gives the following tensor contractions:
\begin{equation}
    \!M_2(\ket{\psi(A)})=-\log \left[2^N\;\frac{\raisebox{-0.45\height}{\includegraphics[height=5 em]{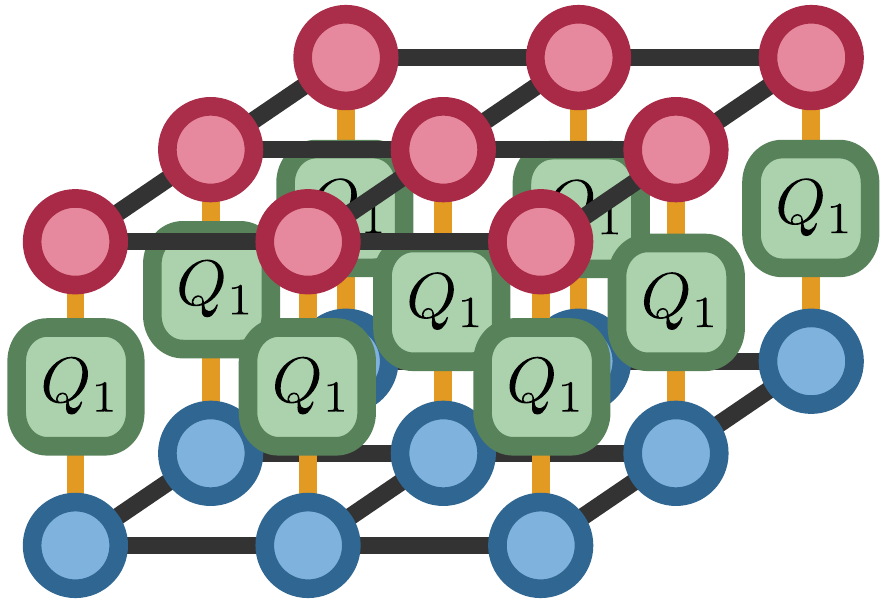}}}{\left(\raisebox{-0.45\height}{\includegraphics[height=5 em]{figures/double_layer_network.pdf}}\right)^4}\right]
    \label{eq:sre_PEPS}
\end{equation}
with
$
    \raisebox{-0.38\height}{\includegraphics[height=3.3 em]{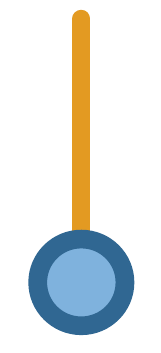}}=\raisebox{-0.33\height}{\includegraphics[height=4 em]{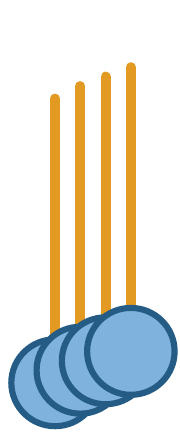}}
$
representing the 4-copy replicated tensors. It is well known that these contractions can be efficiently carried out in the 1D MPS case~\cite{Haug_Piroli_2023}. However, it is not entirely clear whether an efficient algorithm for 2D PEPS can exist: the denominator contraction in Eq.~\eqref{eq:sre_PEPS} is known to be computationally hard in the worst case for PEPS~\cite{Schuch_Wolf_Verstraete_Cirac_2007}, but the SE contraction has additional structure: first, it is a ratio between two PEPS contractions that are dependent on one another, so simplifications cannot be excluded, and second, the $Q_1$ projector at each physical site restricts the Hilbert space accessible to the 4-copy replicated tensor states. Thus, an efficient algorithm cannot be ruled out by the hardness of PEPS contractions alone. Despite this, we are able to prove that the task of computing SE of PEPS retains the full hardness of computing generic 2D contractions.

\begin{theorem}[SE evaluation is $\sharpP-$hard]
        \label{main:thm_sp}
        Problem~\ref{main:prob_sp} for $D=4$ and digit precision $m=\Omega(\poly(N)+\log \alpha)$ is $\sharpP$-hard. It can be solved with polynomial-time post-processing and a single call to a $\sharpP$ oracle.
\end{theorem}
Recall that $\sharpP$ is the class of counting problems that can be represented as satisfying assignments of an (efficiently computable) polynomial-size Boolean circuit~\cite{valiant_permanent}. Thus, $\sharpP$ is harder than $\NP$. Here, $\sharpP-$hardness of Problem~\ref{main:prob_sp} means that any other problem in $\sharpP$ can be reduced with a polynomial number of operations to the computation of the SE of a PEPS. To prove this statement, one has to find a way to map counting the YES instances of a Boolean circuit into computing SE for a suitable PEPS. It is known that counting YES instances of a Boolean circuit can be written as a TN contraction where the tensors enforce the truth tables of the classical gates~\cite{Kourtis_Chamon_Mucciolo_Ruckenstein_2019}. Given this TN one can then define a PEPS of bond dimension 4 that generates the state
 \begin{equation} \label{eq:form_marked_qubit}
          \frac{t\ket{0}+\ket{1}}{\sqrt{1+t^2}}\otimes \ket{0\dots0}\,,
    \end{equation}
where $t$ is the value of the TN contraction (\cref{fig:methods}). Then the SE of said state will be in 1:1 correspondence with the original counting, and with polynomial-time post-processing one is able to reconstruct its value. A complete proof of Theorem~\ref{main:thm_sp} can be found in App.~\ref{sec:sp-complexity}.

\textit{Approximate SE evaluation---}The solution of Problem~\ref{main:prob_sp} retains computational hardness also when trying to computing PEPS SE up to constant accuracy: in particular, our second result establishes that solving Problem~\ref{main:prob_sp} up to a constant error is computationally hard for both classical and quantum computers, even for small bond dimension. In particular, if a polynomial-time classical algorithm for deciding this problem existed, then both the polynomial hierarchy would collapse at the first level (namely, $\mathrm{P}=\NP$ would follow) and polynomial-time quantum computation, i.e. $\BQP$, would collapse into $\rm P$. Moreover, if a polynomial-time quantum algorithm for Problem~\ref{main:prob_sp} up to constant error existed, it would also imply that quantum computers can solve every $\NP$ problem; both implications are highly unlikely under standard complexity assumptions. To prove these implications, we show that Problem~\ref{main:prob_sp} is hard for a complexity class known as $\CeqP$~\cite{gapP,FortnowRogers1999ComplexityLimitations}
which consists of the task of determining whether two $\sharpP$ functions yield exactly the same count.
\begin{theorem}[Constant-accuracy SE is hard]
\label{main:cor_constant_accuracy}
For every fixed integer $\alpha>1$, solving Problem~\ref{main:prob_sp} up to additive error at most $\varepsilon_\alpha = \Omega(1)$, independent of $N$, is
$\CeqP$-hard, already
at $D=33$.
For $\alpha=2$, one may take $\varepsilon_2=1/100$.
\end{theorem}

Note the tradeoff between the strength of the hardness result and the required accuracy: here we obtain only $\CeqP$-hardness, but at constant accuracy, whereas \cref{main:thm_sp} earlier establishes $\sharpP$-hardness at the cost of a much more stringent approximation. Regardless, this suffices, under the standard complexity assumptions discussed earlier, to rule out an efficient classical or quantum algorithm for Problem~\ref{main:prob_sp}. To map a generic $\CeqP$ cancellation problem to an constant accuracy SE evaluation instance, one follows a similar procedure to the one used to prove Theorem~\ref{main:thm_sp}. This time one can write two Boolean formulas as classical Boolean circuits and compute two $\sharpP$ functions $f_1$ and $f_2$ counting their satisfying instances with two TN contractions $\mathcal{C}(A_{f_1})$ and $\mathcal{C}(A_{f_2})$. Then by computing $\mathcal{C}[A_{f_1}\oplus (-A_{f_2})]=\mathcal{C}(A_{f_1})-\mathcal{C}(A_{f_2})$, one is able to write $f_1-f_2$ in a single TN contraction. Then, using these tensors one is able to build a PEPS state containing $\mathcal{C}[A_{f_1}\oplus (-A_{f_2})]$ in the amplitude of one of its qubits. By construction, evaluating said PEPS SE up to constant error allows for deciding if $f_1$ and $f_2$ cancel each other out or if their difference is greater than zero. We refer to \cref{sec:constant_acc_se} of the Appendix for details about the construction and proof.

\textit{Exact stabilizer membership}---Despite having established that quantifying nonstabilizerness of a PEPS is hard, it is often enough to decide whether a given state is a stabilizer state at all, i.e. to answer the stabilizer membership problem. This seemingly easier question is nonetheless relevant: stabilizer states are far more efficient to prepare and manipulate on quantum devices~\cite{DehaeneDeMoor2003, Maslov2018}, possess a structured entanglement spectrum~\cite{Fattal_Cubitt_Yamamoto_Bravyi_Chuang_2004,Iannotti_2025} and admit efficient evaluation of physical quantities such as expectation values of observables, which remain computationally hard to extract from generic PEPS~\cite{Schuch_Wolf_Verstraete_Cirac_2007}. We formalize this as follows.

\begin{problem}[Exact PEPS stabilizer membership --  $\stabpeps$]\label{main:prob_exact-membership}
  Given a non-zero PEPS $A=\{A_i\}$ of bond dimension $D$ on the $N-$qubit square lattice with rational entries, output $\rm YES$ if and only if the corresponding normalized state $\ket{\psi(A)}$ belongs to $\STAB_N$, otherwise output $\rm NO$.
\end{problem}
While one could imagine this weaker question is efficiently solvable, our next result instead proves that deciding exact PEPS stabilizer membership is still computationally hard for classical and quantum alike, by placing it as a complete problem for the previously mentioned $\CeqP$ complexity class.

\begin{theorem}[$\stabpeps$ is $\CeqP$-complete]
\label{main:thm_exact-membership}Problem~\ref{main:prob_exact-membership} with $D= 33$ is $\CeqP$-complete under polynomial-time reductions.
\end{theorem}

Completeness means that Problem~\ref{main:prob_exact-membership} can be solved by a $\CeqP$ oracle and, at the same time, any $\CeqP$ problem can be mapped to deciding if a PEPS is a stabilizer state. The $\CeqP$ complexity class arises naturally in this problem: to satisfy the $M_\alpha(\psi)=0$ stabilizer criterion, an exact cancellation is necessary between two $\sharpP$ functions, namely the argument of the logarithm in Eq.~\eqref{main:se_def} and the norm of the PEPS itself which is also proven to be a $\sharpP$ function in Ref.~\cite{Schuch_Wolf_Verstraete_Cirac_2007}. Then, mapping a generic $\CeqP$ problem in an instance of Problem~\ref{main:prob_exact-membership} requires the same encoding construction used in the proof of~\cref{main:thm_constant-threshold} to encode a the difference of two $\sharpP$ counts in the single qubit amplitude of a PEPS, and then testing for its stabilizer membership allows to know if the two counts cancel each other. 

\textit{Approximate stabilizer membership decision}---Since exactly determining whether a PEPS is a stabilizer is computationally hard, we consider a relaxed version of the previous question, namely whether one can learn if the distance between the generated PEPS and the set $\STAB_N$ is above or below a certain threshold. This question can be formulated as a promise problem:
\begin{problem}[Approximate PEPS stabilizer membership $\ApproxSTABPEPS_{a,b}$]
\label{main:prob_approx-membership}
For $0\le a<b$, given as input a rational-entry $N$-qubit PEPS $A=\{A_i\}$ on the square lattice with the promise
\begin{equation}
\begin{split}
        \yes: &\quad \dist(\ket{\psi(A)},\STAB_N)\le a,\\
        \no:  &\quad \dist(\ket{\psi(A)},\STAB_N)\ge b,
\end{split}
\end{equation}
decide which is the case. Here
$\ket{\psi(A)}$ represents the normalized state described by the PEPS input, and $\dist(\ket{\psi(A)},\STAB_N)$ is the minimum Euclidean distance between $\ket*{\psi(A)}$ and an $N-$qubit pure stabilizer state.
\end{problem}
In hardness proofs for promise problems, the relevant parameter is the promise gap~\cite{Goldreich2006Promise}, namely $b-a$.
In our case, we prove that $\ApproxSTABPEPS_{a,b}$ is still $\CeqP-$hard for \emph{constant} promise gaps, and it is $\CeqP-$complete when $a=0$. This means that also approximately determining stabilizer membership for PEPS still retains the full hardness of the exact membership problem, entailing the same complexity-class collapse discussed earlier.
\begin{theorem}[Constant gap $\ApproxSTABPEPS_{a,b}$ $\CeqP-$hardness]
\label{main:thm_constant-threshold}
For every $0\le a<b\le 1/5$, the promise problem $\ApproxSTABPEPS_{a,b}$ is $\CeqP$-hard. In particular, $\ApproxSTABPEPS_{0,1/16}$ is $\CeqP$-complete as a promise problem.
\end{theorem}
To prove this theorem, we again write the difference of two $\sharpP$ functions as two TN contractions and then encode their difference in the single qubit amplitude of a PEPS of the form of Eq.~\eqref{eq:form_marked_qubit}. Then, exact cancellation of the $\sharpP$ counts corresponds to an upper bound of $a$ on the distance from the stabilizer set, whereas lack of exact cancellation corresponds to a $b-$lower bound on the same distance: these are the NO instances of the promise problem. Again, we refer to App.~\ref{sec:approx-membership} for a detailed proof. 

\textit{Local-stabilizer equivalence}---A related problem is deciding if can a PEPS be mapped to a stabilizer state using factorized unitaries with respect to a given bipartition. We denote as $\lstab$ the set of states that can be mapped to a stabilizer state by use of unitaries of the form $U_A\otimes U_B$, and formulate this question as a decision problem:
\begin{problem}[$\lstab$ membership]
\label{main:prob_zero-nlm}
Given a rational-entry tensor description $T$ of a nonzero $N$-qubit PEPS on a square lattice
and a specified bipartition $A|B$ of its physical qubits, output YES if $\ket{\psi(T)}\in\lstab$, and NO otherwise.
\end{problem}
It is known that this problem amounts to deciding if the entanglement spectrum is flat, with rank an integer power of two~\cite{Fattal_Cubitt_Yamamoto_Bravyi_Chuang_2004,sierant2026exactquantificationnonlocalmagic,cao2025gravitational}. Note that, hardness of extracting the entanglement spectrum of a PEPS alone~\cite{cheng2019ubiquitouscomplexityentanglementspectra} does not suffice to show hardness of deciding Problem~\ref{main:prob_zero-nlm}.
\begin{theorem}[Hardness of $\lstab$ membership]
\label{main:thm_nlm-hardness}
Problem~\ref{main:prob_zero-nlm} is $\CeqP$-hard under polynomial-time postprocessing for $D=33$.
\end{theorem}
The proof is straightforward by employing a similar strategy as previously: one can encode a difference of $\sharpP$ counts $t$ in a TN contraction, and then in a PEPS proportional to $(t\ket{00}+\ket{11})\otimes\ket{0}\ldots \ket{0}$. Knowledge of membership of such state to $\lstab$ allows to to reconstruct the value of $t$. Details on the proof can be found in App.~\ref{subsec:nonlocal-magic}


\textit{Discussion}---We have shown that nonstabilizerness of 2D $N-$qubit PEPS is worst-case intractable: computing the SE of such states, even at small bond dimension, is as hard as generic tensor-network contractions, which are intractable for both classical and quantum computers with polynomial resources, subject to standard complexity assumptions. This hardness persists for the seemingly simpler problem of deciding stabilizer membership of a PEPS, both approximately and exactly, or deciding its stabilizer membership up to a local unitary equivalence along a given bipartition.

The constant-accuracy hardness established in Theorem~\ref{main:cor_constant_accuracy} shows that relaxing numerical evaluation to a sufficiently small fixed additive tolerance does not remove the worst-case obstruction to estimating PEPS nonstabilizerness. This result constrains algorithms with guarantees for arbitrary PEPS, while leaving room for efficient methods on physically motivated subclasses or typical instances encountered in numerical work. An open question is whether constant-accuracy estimation also retains the $\sharpP$-hardness of exact SE evaluation, beyond the $\CeqP$-hardness established here.

Several other directions remain open.  It would be important to understand how the present hardness changes under additional physical constraints, such as injectivity, translational invariance, the existence of a uniformly gapped parent Hamiltonian and its possible symmetries.  Since the reductions already use constant physical dimension and constant bond dimension, any efficient subclass must rely on structure beyond bounded local dimensions.  Further extensions include other two-dimensional ans\"atze such as isometric TN, and other nonstabilizerness monotones or stabilizer-based witnesses.  Clarifying which of these quantities remain hard, and which become tractable under physically natural restrictions, would sharpen the role of nonstabilizerness in tensor-network simulation of quantum many-body systems. 

\textit{Acknowledgments}---We thank Rahul Trivedi and Daniele Iannotti for helpful discussions. G.S. acknowledges funding from the Deutsche Forschungsgemeinschaft (DFG, German Research Foundation) under Germany's Excellence Strategy Grant No. EXC2111-390814868.

The authors acknowledge use of AI for parts of this project. In particular, OpenAI’s language models 5.6 Sol and 6 Astra were used by the authors for refinements of the proofs of Theorems 2, 4, and 5, as well as general proofreading and typographical corrections. The ideas of the proofs are human-generated.

\makeatletter
\let\GE@bibsectionwithtoc\bibsection
\renewcommand{\bibsection}{\let\GE@savedaddcontentsline\addcontentsline
  \let\addcontentsline\@gobblethree
  \GE@bibsectionwithtoc
  \let\addcontentsline\GE@savedaddcontentsline
}
\makeatother

\bibliographystyle{apsrev4-2}
\bibliography{refs_supplemental}

\clearpage 
\appendix
\onecolumngrid
\makeatletter
\let\set@footnotewidth\set@footnotewidth@one
\makeatother
\begin{center}
{\large{\bf Supplemental Material for\\
 ``Nonstabilizerness of quantum tensor network states is intractable in two dimensions''}}
\end{center}

\tableofcontents
\medskip\medskip\medskip

\section{Tensor Networks, PEPS, and Boolean counting}
\label{sec:tensor-networks}

This section contains all the concepts relative to Tensor Networks (TN), and specifically PEPS, that are relevant to the results of this work. A comprehensive reference to the subject can be found in \cite{Cirac_Perez-Garcia_Schuch_Verstraete_2021}. 
\begin{definition}[Scalar tensor network]
{
A scalar TN $T$ on a graph $\Lambda=(V,E)$ assigns to each vertex
$v\in V$ a tensor
\begin{equation}
  (T_v)_{(\alpha_e)_{e\rightarrow v}}\in\mathbb{C},
  \qquad \alpha_e\in\{1,\dots,D_e\},
\end{equation}
with one index $\alpha_e$ for each edge $e$ incident to $v$. Its \emph{contraction} is
\begin{equation}
  \cont(T):=\sum_{\{\alpha_e\}_{e\in E}} \prod_{v\in V}(T_v)_{(\alpha_e)_{e\rightarrow v}},
\end{equation}
where the sum runs over all $\alpha_e\in\{1,\dots,D_e\}$.
}
\end{definition}

For rational-entry tensors, $\cont(T)$ is a rational number, and after multiplying by the maximum common denominator it is an integer-valued signed sum over internal index assignments.

We now briefly introduce the PEPS formalism and the conventions used throughout this Supplemental Material.
\begin{definition}[Projected Entangled Pair States (PEPS)]
\label{PEPS:def}
Let $\Lambda=(V,E)$ be a finite square or rectangular lattice, denoting by $N:=|\Lambda|$ the number of physical sites. To each site we associate a qubit (\emph{physical dimension} $d=2$) and the total Hilbert-space is thus $(\mathbb C^2)^{\otimes N}$. At each edge $e\in E$ we assign a \emph{virtual space} {$\mathbb C^{D_e}$} where each \emph{bond dimension} $D_e$ may depend on the edge. We place on each edge an (unnormalized) maximally entangled vector
\begin{equation}
    \ket*{\Omega_e}
    :=
    \sum_{\alpha_e=1}^{D_e}
    \ket*{\alpha_e}\ket*{\alpha_e}.
    \label{eq:virtual-pair}
\end{equation}
At every site $v\in V$ we then apply a local linear map
\begin{equation}
     A_v:
    \bigotimes_{e\rightarrow v}\mathbb C^{D_e}
    \longrightarrow
    \mathbb C^2,
    \label{eq:local-peps-map}
\end{equation}
whose output is the physical qubit at $v$. The corresponding unnormalized PEPS is
\begin{equation}
    \ket*{\widetilde\psi(A)}
    =
    \left(\bigotimes_{v\in V}A_v\right)
    \left(\bigotimes_{e\in E}\ket*{\Omega_e}\right).
    \label{eq:peps-map-definition}
\end{equation}
\end{definition}

After choosing bases on every virtual space, the local map~\eqref{eq:local-peps-map} is equivalently represented by a tensor
\begin{equation}
{
    \left(A_v^{s_v}\right)_{(\alpha_e)_{e\rightarrow v}}\in\mathbb{C},
    \qquad
    s_v\in\{0,1\},
    \qquad
    \alpha_e\in\{1,\ldots,D_e\}.
}
\end{equation}
The index $s_v$ is called the \textit{physical} index, since it labels the computational basis of the physical qubit at site $v$. The \emph{bond dimension} of a PEPS is defined as $D = \max_{e \in E} D_e$.
For every computational basis configuration
$s=(s_v)_{v \in V}\in\{0,1\}^N$,
the corresponding amplitude of the PEPS is obtained by contracting all virtual indices,
\begin{equation}
{\braket*{s}{\widetilde\psi(A)} = \;}
    \cont_A(s)=\sum_{\{\alpha_e\}_{e\in E}}
    \prod_{v \in V}
    \left(A_v^{s_v}\right)_{(\alpha_e)_{e\rightarrow v}}\,.
    \label{eq:peps-coordinate-amplitude}
\end{equation}
Thus each amplitude $\braket*{s}{\widetilde\psi(A)}$ of a PEPS corresponds to a scalar TN \textit{contraction}. In representing tensor contractions, a graphical notation is often used for simplicity: each circle represents a single tensor, black legs represent virtual indices and yellow legs indicate physical indices. Every connected leg between two tensors is to be interpreted as the corresponding index between the two tensors being summed over, whereas non-attached legs are to be interpreted as free indices. According to this notation, the state vector then reads
\begin{equation}
    \ket*{\widetilde\psi(A)}
    =
    \sum_{s\in\{0,1\}^N}
    \cont_A(s)\ket*{s}=\!\!\!\raisebox{-0.4\height}{\includegraphics[height=4 em]{figures/single_layer_network.pdf}}
\end{equation}
This coordinate representation is the one employed throughout the reductions developed below.

Throughout this work, we explicitly distinguish between normalized and unnormalized PEPS. The squared norm of the unnormalized vector is given by the following scalar tensor contraction:
\begin{equation}
\begin{split}
    Z(A)
    &:=
    \braket*{\widetilde\psi(A)}{\widetilde\psi(A)}
    =
    \sum_{s\in\{0,1\}^N}
    \bigl|\cont_A(s)\bigr|^2=\!\!\!\raisebox{-0.33\height}{\includegraphics[height=5 em]{figures/double_layer_network.pdf}}
    \label{eq:peps-norm}
    \end{split}
\end{equation}
where red circles represent the conjugated tensors $\overline{A_v^{s}}$, and whenever $Z(A)>0$ we define
\begin{equation}
    \ket*{\psi(A)}
    =
    \frac{\ket*{\widetilde\psi(A)}}{\sqrt{Z(A)}}.
\end{equation}
Since every physical quantity regarding PEPS can be expressed as a TN contraction~\cite{Cirac_Perez-Garcia_Schuch_Verstraete_2021}, the graphical notation can be used to represent every relevant quantity; in particular, it can be used to express expectation values of a local observable $O_R$:
\begin{equation}\label{eq:expval}
    \mel*{\widetilde{\psi}(A)}{O_R}{\widetilde{\psi}(A)}=\raisebox{-0.3\height}{\includegraphics[height=5 em]{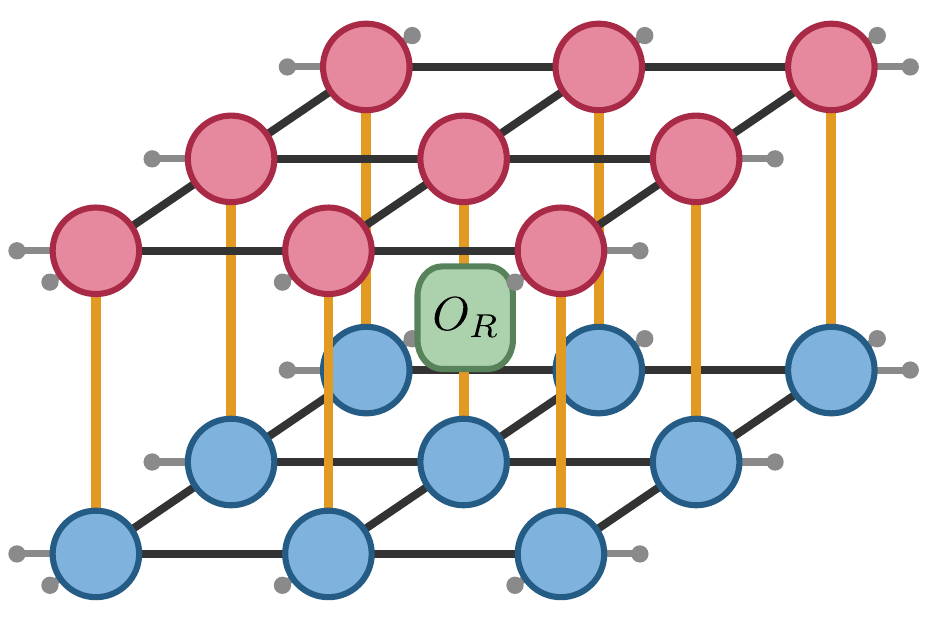}}
\end{equation}

\section{Complexity-theoretic preliminaries}
\label{sec:complexity}
In this section we introduce the relevant complexity theory concepts that we employ in this work. A comprehensive review of these topics is provided in~\cite{AroraBarak2009,Papadimitriou1994}.
\subsection{Reductions}

For decision problems $A,B\subseteq\{0,1\}^*$, a \textit{polynomial-time many-one reduction} $A\le_m^P B$ is a polynomial-time computable map $R$ such that
\begin{equation}
        x\in A
        \quad\iff\quad
        R(x)\in B.
\end{equation}
If $A\le_m^P B$, then an efficient algorithm for $B$ would imply an efficient algorithm for $A$.  For promise problems $\Pi=(\Pi_{\yes},\Pi_{\no})$ and $\Gamma=(\Gamma_{\yes},\Gamma_{\no})$, a \textit{promise many-one reduction} is a polynomial-time map $R$ satisfying
\begin{equation}
        x\in\Pi_{\yes}\Rightarrow R(x)\in\Gamma_{\yes},
        \qquad
        x\in\Pi_{\no}\Rightarrow R(x)\in\Gamma_{\no},
\end{equation}
with no requirement outside the promise~\cite{Goldreich2006Promise}.

For exact function problems we use the following standard exact postprocessing notion.

\begin{definition}[Exact function reduction]
\label{def:function-reduction}
Let $F$ and $G$ be exact function problems.  We say that $F$ \textit{reduces to} $G$ \textit{under polynomial-time exact postprocessing} if there are polynomial-time computable maps $R$ and $P$ such that
\begin{equation}
        F(x)=P\bigl(x,G(R(x))\bigr)
\end{equation}
for every input $x$.  If $P$ is the identity, the reduction is \emph{parsimonious}.  If $P$ consists only of known exact arithmetic operations, such as clearing denominators, multiplying by known factors, or inverting an efficiently computable injective rational map on a discrete range, we will call the reduction \emph{weakly parsimonious}.
\end{definition}

\subsection{Counting and equality classes}
In this subsection we introduce and define the counting and decision complexity classes that are relevant to this work.
\begin{definition}[$\sharpP$ functions]
    A function $f:\{0,1\}^*\to\mathbb N$ lies in $\sharpP$ if it can be represented as the number of satisfying assignments of a polynomially evaluable Boolean formula~\cite{valiant_permanent}, or, equivalently, a polynomial-size Boolean circuit $C_x(y)$~\cite{AroraBarak2009} on the condition that the map $x\mapsto C_x$ is efficiently computable.
\end{definition}

In this definition we actually used the notion of a $\sharpP-$complete problem, namely $\# \rm SAT$, in place of the standard definition involving counting the accepting paths of a non-deterministic polynomial-time Turing machine, but the two are equivalent \cite{Karp_1972}.
\begin{definition}[$\GapP$ functions]
    The class $\GapP$ is the closure of $\sharpP$ under subtraction~\cite{gapP}.  Thus $g:\{0,1\}^* \to \mathbb{Z}$ is in $\GapP$ if and only if there exist $f_+,f_-\in\sharpP$ such that
    \begin{equation}
            g(x)=f_+(x)-f_-(x).
    \end{equation}
\end{definition}
    GapP functions may be positive, negative, or zero.  They are closed under addition, subtraction, and multiplication.  In particular, if $g\in\GapP$, then $g^2\in\GapP$~\cite{FortnowRogers1999ComplexityLimitations}. 

\begin{definition}[$\CeqP$ problems]
The decision class $\CeqP$ (pronounced ``C-equal P'') consists of exact zero tests of $\GapP$ functions:
\begin{equation}
        L\in\CeqP
        \quad\iff\quad
        \exists\, g\in\GapP\text{ such that }
        x\in L\iff g(x)=0.
        \label{eq:ceqp-definition}
\end{equation}
\end{definition}

It is useful to contrast this complexity class with $\PP$, which consists of sign tests $g(x)>0$ for $\GapP$ functions.  Thus $\CeqP$ is naturally associated with exact cancellations, while $\PP$ is associated with majority or threshold tests.  Standard background on these classes can be found in Refs.~\cite{AroraBarak2009,Papadimitriou1994,Sipser2013,gapP,FortnowRogers1999ComplexityLimitations}.

Because $\sharpP$ and $\GapP$ are function classes while $\CeqP$ is a decision class, statements such as $\CeqP\subseteq\sharpP$ are not formally meaningful.  The correct relation used subsequently is that equality to zero of an explicitly constructed GapP function yields a $\CeqP$ decision problem, while exact evaluation of rational tensor contractions belongs to $\GapP$, computable with a  $\sharpP$ oracle and polynomial-time exact arithmetic.

\subsection{Computational significance of \texorpdfstring{$\CeqP$}{C=P}-hardness and consequences of an efficient algorithm for 
\texorpdfstring{$\CeqP$}{C=P}}
\label{subsec:ceqp-significance}

The class $\CeqP$ is not frequently invoked in quantum information. We thus collect below complexity-theoretic evidence that it is highly unlikely to be easy to decide, by exploring the consequences of the hypothetical existence of a polynomial-time algorithm for it.

The significance of $\CeqP$ is best understood from its defining $\GapP$ representation.  Writing $g=f_{+}-f_{-}$ with $f_{+},f_{-}\in\sharpP$, each term may count an exponentially large family of accepting witnesses, while the decision predicate asks whether the two counts agree exactly.  A YES instance can therefore arise from a perfect cancellation between two exponentially large integers.

The relation to more familiar classes can be summarized by the standard containments
\begin{equation}
    \mathrm{coNP}\subseteq\CeqP\subseteq\PP\subseteq\mathrm{PSPACE},
    \qquad
    \mathrm{NP}\subseteq\mathrm{co}\text{-}\CeqP\subseteq\PP.
    \label{eq:ceqp-containments}
\end{equation}
The first inclusion follows already from Boolean satisfiability: if $s(\varphi)$ is the number of satisfying assignments of a formula $\varphi$, then $s\in\sharpP\subseteq\GapP$ and $\varphi$ is unsatisfiable if and only if $s(\varphi)=0$.  Since \textsc{Unsat} is $\mathrm{coNP}$-complete, every problem in $\mathrm{coNP}$ reduces to a $\CeqP$ zero test.  For the upper inclusion, if $x\in L\iff g(x)=0$ with $g\in\GapP$, then closure of $\GapP$ under products and subtraction gives $h(x):=1-g(x)^2\in\GapP$, and $x\in L\iff h(x)>0$; hence $L\in\PP$.  Likewise, $g(x)\neq0$ is equivalent to $g(x)^2>0$, which places the complementary nonzero-test class in $\PP$ and contains the usual $\mathrm{NP}$ witness-existence predicate.  

We now record the complexity-theoretic consequences of the hypothetical existence of an exact
polynomial-time algorithm for $\CeqP$. Recall that $\CeqP$ is understood as the class of languages
$L$ for which there exists a function $g\in\GapP$ satisfying
\begin{equation}
    x\in L\quad\iff\quad g(x)=0.
\end{equation}
Thus, the hypothesis that $\CeqP$ can be solved in
polynomial time means that every language in $\CeqP$
has an exact deterministic polynomial-time decision algorithm. Equivalently,
it is sufficient that a $\CeqP$-complete problem under
polynomial-time reductions belong to $\mathrm{P}$.

\begin{lemma}
\label{lem:np-coceqp}
The following containments hold:
\begin{enumerate}
    \item 
    $\mathrm{NP}\subseteq\operatorname{co}\CeqP$,  $\mathrm{coNP}\subseteq\CeqP$.

    \item $\PP\subseteq\mathrm{NP}^{\CeqP}$.

\end{enumerate}

\end{lemma}

\begin{proof}

\begin{enumerate}
\item Let $L\in\mathrm{NP}$, and let $V$ be a polynomial-time verifier for
$L$. After padding the witness register if necessary, define
\begin{equation}
    f(x)=\#\left\{w\in\{0,1\}^{q(|x|)}:V(x,w)=1\right\},
\end{equation}
where $q$ is a polynomial. By definition, $f\in\sharpP$ and hence
$f\in\GapP$. Moreover,
\begin{equation}
    x\in L\quad\iff\quad f(x)\neq 0.
\end{equation}
It follows that $L\in\operatorname{co}\CeqP$, proving
the first containment. Taking complements gives
$\mathrm{coNP}\subseteq\CeqP$.
\item This is an immediate consequence of Ref.~\cite{Toran1991CountingQuantifiers}, Corollary 3.11 and Theorem 4.1, which together give $\rm{NP^{PP}=NP^{\CeqP}}$.
\end{enumerate}
\end{proof}

The next theorem elucidates the complexity collapse that a hypothetical polynomial algorithm that decides $\CeqP$ problems would imply.

\begin{theorem}
\label{thm:ceqp-collapse}
If
\begin{equation}
    \CeqP
    \subseteq
    \mathrm{P},
\end{equation}
then
\begin{equation}
    \mathrm{P}
    =
    \mathrm{NP}
    =
    \mathrm{coNP}
    =
    \mathrm{PH}
    =
    \PP.
\end{equation}
\end{theorem}

\begin{proof}
The reverse containment
$\mathrm{P}\subseteq\CeqP$ is immediate: for any
$L\in\mathrm{P}$, the polynomial-time computable function
\begin{equation}
    g_L(x)
    =
    \begin{cases}
        0, & x\in L,\\
        1, & x\notin L,
    \end{cases}
\end{equation}
belongs to $\GapP$ and characterizes $L$ by its zero set.
Therefore the hypothesis is equivalent to
\begin{equation}
    \CeqP=\mathrm{P}.
\end{equation}

Since $\mathrm{P}$ is closed under complement, the hypothesis also gives
\begin{equation}
    \operatorname{co}\CeqP
    =
    \mathrm{P}.
\end{equation}
Lemma~\ref{lem:np-coceqp} then implies
\begin{equation}
    \mathrm{NP}
    =
    \mathrm{coNP}
    =
    \mathrm{P}.
\end{equation}
The polynomial hierarchy therefore collapses at P:
\begin{equation}
    \mathrm{PH}=\mathrm{P}.
\end{equation}

Finally, Lemma~\ref{lem:np-coceqp} gives
\begin{equation}
    \PP
    \subseteq
    \mathrm{NP}^{\CeqP}
    =
    \mathrm{NP}^{\mathrm{P}}
    =
    \mathrm{NP}
    =
    \mathrm{P}.
\end{equation}
Together with the trivial containment $\mathrm{P}\subseteq\PP$,
this proves $\PP=\mathrm{P}$.
\end{proof}

For the quantum case, we have:

\begin{corollary}
\label{cor:quantum-collapse}
The hypothesis
$\CeqP=\mathrm{P}$ implies
\begin{equation}
    \mathrm{P}=\mathrm{BPP}=\mathrm{BQP}=\PP=\mathrm{postBQP}.
\end{equation}
\end{corollary}

\begin{proof}
The standard unconditional containments are~\cite{FortnowRogers1999ComplexityLimitations}
\begin{equation}
    \mathrm{P}\subseteq\mathrm{BPP}\subseteq\mathrm{BQP}\subseteq\PP\,.
\end{equation}
Since Theorem~\ref{thm:ceqp-collapse} gives
$\PP=\mathrm{P}$, these containments imply
\begin{equation}
    \mathrm{P}
    =
    \mathrm{BPP}
    =
    \mathrm{BQP}
    =
    \PP.
\end{equation}
Finally, since
$\mathrm{postBQP}=\PP$~\cite{Aaronson_2005}, then
$\mathrm{postBQP}=\mathrm{P}$ .
\end{proof}

Moreover, under the weaker assumption that a polynomial-time quantum algorithm that decides $\CeqP$ problems exists, the following collapse, highly unlikely under standard complexity assumptions, would also follow.

\begin{theorem}\label{thm:ceqp-np-bqp}
If $\CeqP \subseteq \BQP$, then $\NP \subseteq \BQP$.
\end{theorem}

\begin{proof}
Eq.~\eqref{eq:ceqp-containments} gives $\coNP \subseteq \CeqP \subseteq \BQP$.
Taking complements and using closure of $\BQP$ under complement, we obtain $\NP \subseteq \BQP$.
\end{proof}

\section{Generic PEPS contraction hardness}
One of the central strengths of the PEPS formalism is that it provides an exponentially compressed description of many-body quantum states whenever the bond dimension scales favorably with the system size~\cite{Cirac_Perez-Garcia_Schuch_Verstraete_2021}. Although to specify a generic $N$-qubit vector $2^N$ complex amplitudes are required, a PEPS is specified only by its local tensors. For fixed physical dimension and bond dimension $D = O(\poly N)$, the total number of tensor entries grows only polynomially with the system size (on the condition that each entry has a polynomially bounded bit length representation), thus PEPS constitute an efficiently describable family of quantum states despite representing vectors living in exponentially large Hilbert spaces.

As shown before, the contraction viewpoint naturally encompasses all quantities appearing in this work. Computational basis amplitudes are given by the scalar contractions seen in Eq.~\eqref{eq:peps-coordinate-amplitude}. The norm corresponds to the double-layer TN obtained by contracting a copy of the PEPS with its conjugate shown in Eq.~\eqref{eq:peps-norm}. More generally, expectation values of local observables are obtained by inserting the observable on the corresponding physical legs before performing the double-layer contraction as seen in Eq.~\eqref{eq:expval}, where $R\subseteq\Lambda$ denotes the support of the observable. 

Despite their compact description and simplicity in their graphical representation, contracting arbitrary two-dimensional PEPS is computationally intractable in the worst case, even restricting to PEPS with bond dimension $D=2$ on the square lattice. In particular, Schuch, Wolf, Verstraete, and Cirac proved that exact PEPS norm evaluation, as well as local expectation-value computation, are $\sharpP$-complete problems~\cite{Schuch_Wolf_Verstraete_Cirac_2007}; see also the modern review~\cite{Cirac_Perez-Garcia_Schuch_Verstraete_2021}. 
We will use the following consequence of the PEPS contraction hardness.

\begin{theorem}[PEPS contraction hardness~\cite{Schuch_Wolf_Verstraete_Cirac_2007}]
\label{thm:peps-contraction-hardness}
Exact contraction of PEPS, including exact norm evaluation
\begin{equation}
        Z(A)=\braket*{\widetilde\psi(A)}{\widetilde\psi(A)},
\end{equation}
is $\sharpP$-hard under exact function reductions.  Conversely, for rational-entry PEPS, the relevant exact contractions are computable in $\sharpP$ after clearing denominators.
\end{theorem}

Notice that while \cite{Schuch_Wolf_Verstraete_Cirac_2007} proves that PEPS contractions are $\sharpP-$complete in the worst case, \cite{Haferkamp_Hangleiter_Eisert_Gluza_2020} proves that these contractions are also hard in the average case if one wants to obtain the contraction value up to multiplicative precision. 

\section{Stabilizer entropy preliminaries}
\label{sec:stabilizer-prelims}
This section contains the concepts and properties of stabilizer states and stabilizer R\'enyi entropy (SE) that are relevant for this work. For a more comprehensive review, refer to~\cite{Leone_Oliviero_Hamma_2022,nielsen_chuang}.

Let $\pauli_N:=\{I,X,Y,Z\}^{\otimes N}$ denote the set of $N$-qubit Pauli strings modulo phases.  Let $\STAB_N$ denote the set of pure $N$-qubit stabilizer states, up to global phase.  Given a pure normalized $N-$qubit state $\ket*\psi$, define its \textit{characteristic distribution} as 
\begin{equation}\label{eq:chidist}
    \Xi(\psi):=\{2^{-N}\Tr[P\ketbra{\psi}]^2\}_{P\in \pauli_N}\,.
\end{equation}
It is straightforward to see that this is a probability distribution over $\pauli_N$, i.e. $\Xi_P(\psi)\geq 0 \;\forall P\in \pauli_N
$ and $\sum_P \Xi_P(\psi)=1$ for all pure normalized states. 

{A direct way to quantify nonstabilizerness of pure states is via the moments of the distribution $\Xi(\psi)$. This relies on the intuition, formalized later, that pure stabilizer states correspond to maximally concentrated distributions with a constraint on the cardinality of their support given by the purity condition, while states with nonstabilizerness yield flatter distributions~\cite{tirrito_2024_flatness}. Perhaps the simplest such measure is the \textit{SE} of a pure state $\ket*\psi$, defined} as the (offset) $\alpha-$R\'enyi entropy of its characteristic distribution, namely 
\begin{equation} \label{eq:sre_def}
    M_\alpha(\psi):=\frac{1}{1-\alpha}\log \left(2^{N(\alpha-1)}\sum_P \Xi_P(\psi)^\alpha\right) \;.
\end{equation}
along with its argument, which we denote by $\alpha-$stabilizer purity:
\begin{equation} \label{eq:stp_def}
    \stp_\alpha(\psi)=2^{N(\alpha-1)}\sum_P \Xi_P(\psi)^\alpha
\end{equation}
SE and purities satisfy the following property, which we also prove here for completeness:

\begin{lemma}[Faithfulness of SE/purity~\cite{Leone_Oliviero_Hamma_2022}]
\label{lem:sp-faithfulness}
For a normalized pure $N$-qubit state $\ket*\psi$ and for integer $\alpha>1$,
\begin{equation}
        M_\alpha(\psi)=0\iff\stp_\alpha(\psi)=1
        \quad\iff\quad
        \ket*\psi\in\STAB_N.
\end{equation}
\end{lemma}

\begin{proof}

Let $\psi=\ketbra{\psi}$ be normalized and pure, and put
$a_P=\Tr(P\psi)$. Pauli orthogonality gives
\begin{equation}
\psi=2^{-N}\sum_Pa_PP,\qquad
1=\Tr\psi^2=2^{-N}\sum_Pa_P^2.
\end{equation}
Every $a_P$ is real and $|a_P|\le1$. Therefore
\begin{equation}
\Xi_P=2^{-N}a_P^2\in[0,2^{-N}],\qquad
\sum_P\Xi_P=1.
\end{equation}
For every real $\alpha>1$, the correct normalization is
\begin{equation}
\mathrm{SP}_\alpha
=2^{N(\alpha-1)}\sum_P\Xi_P^\alpha
=2^{-N}\sum_P|a_P|^{2\alpha},
\end{equation}
\begin{equation}
M_\alpha=\frac{\log \mathrm{SP}_\alpha}{1-\alpha}
\end{equation}
For the equality case, put $x_P=a_P^2$. Since $0\le x_P\le1$,
\begin{equation}
\sum_Px_P^\alpha\le\sum_Px_P=2^N,
\end{equation}
with equality if and only if every $x_P$ equals zero or one. Thus
$\mathrm{SP}_\alpha=1$ implies exactly $2^N$ Pauli expectations of
magnitude one. If $|a_P|=1$, zero variance implies
$P|\psi\rangle=a_P|\psi\rangle$. Two anticommuting Hermitian Paulis
cannot have a common eigenvector: applying their products in opposite
orders would give opposite nonzero multiples of that vector.
Consequently all supported Paulis commute.

The signed operators
\begin{equation}
G=\{a_PP:|a_P|=1\}
\end{equation}
form an abelian Pauli subgroup fixing $|\psi\rangle$. It is closed
because products of commuting signed Pauli stabilizers again fix the
same vector; it excludes $-I$, and it has $2^N$ elements. Its common
+1 eigenspace has dimension one. Hence $|\psi\rangle$ is a pure
stabilizer state. Conversely, a pure stabilizer has exactly those
$2^N$ unit-magnitude expectations, so its stabilizer purity is one and its SE is zero.
\end{proof}

{It will be useful for later to ``linearize'' the expression of the $2-$stabilizer purity at the expense of introducing more copies of the state. For that,} define the orthogonal projector
\begin{equation}
        Q_N
        :=
        \frac{1}{4^N}
        \sum_{P\in\pauli_N}P^{\otimes 4}.
        \label{eq:QN-definition}
\end{equation}
This operator factorizes as
\begin{equation}
        Q_N=Q_1^{\otimes N},
        \qquad
        Q_1=\frac14\left(I^{\otimes4}+X^{\otimes4}+Y^{\otimes4}+Z^{\otimes4}\right)\,.
        \label{eq:Q-factorization}
\end{equation}
For a normalized pure state $\ket*\psi$, using $\Tr(A)\Tr(B)=\Tr(A\otimes B)$,  its stabilizer purity can be written as
\begin{equation}
        \stp_2(\psi)
        :=
        2^N\Tr[Q_N\ketbra\psi^{\otimes4}]\,.
        \label{eq:sp-normalized}
\end{equation}
For an unnormalized nonzero state $\ket*{\widetilde{\psi}}$, the stabilizer purity of the corresponding normalized state $\ket\psi$ is
\begin{equation}
        \stp_2({\psi}
        )
        =
        2^N
        \frac{\Tr[
        Q_N\ketbra*{\widetilde{\psi}}^{\otimes4}
        ]}{\braket*{\widetilde{\psi}}^4}\,.
        \label{eq:sp-unnormalized}
\end{equation}
The normalization in \eqref{eq:sp-unnormalized} is essential, as the condition $\stp_2(\psi)=1$ detects stabilizer states only after normalizing the pure state.

\section{Stabilizer entropy evaluation}
\label{sec:sp-complexity}

In this section, we prove $\sharpP-$hardness of computing SE. The formal definition of the computational problem is introduced below: we phrase this problem as the problem of outputting a \textit{dyadic number} (that is, a rational number where the denominator is a power of two) that approximates the value of SE up to $m$ digits.

\begin{problem}[Dyadic evaluation of the $\alpha$-SE]
    \label{prob:dyadic_stabilizer_entropy}
    Fix a constant integer $\alpha>1$. Given the tensor description
    $\{A_i\}$ of an $N$-qubit PEPS on a square lattice and a fixed number of digits $m$ in binary notation, output an integer $k\in\mathbb{Z}_{\geq 0}$ such
    that the dyadic rational
    \begin{equation}
        \widetilde M_\alpha\coloneqq\frac{k}{2^m}
    \end{equation}
    approximates the $\alpha-$SE
     \begin{equation}
        M_\alpha\bigl(\ket*{\psi(A)}\bigr)
        \coloneqq
        \frac{1}{1-\alpha}\log 
        \stp_\alpha\bigl(\ket*{\psi(A)}\bigr)
    \end{equation}
    with additive error
    \begin{equation}
        \left|
        \widetilde M_\alpha-
        M_\alpha\bigl(\ket*{\psi(A)}\bigr)
        \right|
        \leq 2^{-m}.
    \end{equation}
\end{problem}

To prove $\sharpP$ hardness of Problem~\ref{prob:dyadic_stabilizer_entropy}, we start with the simpler case of $\alpha = 2$ stabilizer purity, prove its hardness, and then  generalize to integer $\alpha > 1$. We then establish that computing $m-$digit accurate SE and computing exact stabilizer purity are equivalent problems that can be mapped onto each other using polynomial-time post-processing.

\subsection{\texorpdfstring{$\alpha-$}{}stabilizer purity}

We first formulate the exact $2-$stabilizer purity problem as a rational-valued function problem.

\begin{problem}[Exact PEPS stabilizer-purity evaluation]
\label{prob:sp}
 Let $\Lambda$ be a regular square lattice with $N:=|\Lambda|$ sites, local physical dimension $d=2$, and bond dimension $D$. Given a rational tensor description $\{A_i ^{s_i}\}_{i\in \Lambda,s_i\in\{0,1\}}$ of a nonzero PEPS on this lattice $\ket*{\widetilde{\psi}(A)}$, compute the $2-$stabilizer purity of the corresponding normalized PEPS $\ket{\psi(A)}$:
           \begin{equation}
    \stp_2(\ket{\psi(A)}):=2^N\Tr(\frac{Q_N\ketbra*{\widetilde{\psi}(A)}^{\ot 4}}{\braket*{\widetilde{\psi}(A)}^4})\,.
           \end{equation}
\end{problem}

Our result for its hardness is:

\begin{theorem}[{Exact PEPS stabilizer purity is $\sharpP$-complete}]
\label{thm:sp-hard}
Problem~\ref{prob:sp} for PEPS of bond dimension $D = 4$ is complete for $\sharpP$-type exact counting under weakly-parsimonious reductions. In particular, it is $\sharpP$-hard under polynomial-time exact function reductions and it can be computed with access to a  ${\sharpP}$ oracle and polynomial-time exact rational postprocessing.
\end{theorem}

The proof of this theorem naturally separates in two parts, namely proof of computability with a $\sharpP$ oracle and proof of $\sharpP-$hardness. 
To prove the first statement, we show that one can write the stabilizer purity computation as a PEPS contraction, which is $\sharpP-$complete~\cite{Schuch_Wolf_Verstraete_Cirac_2007}. This part is straightforward; we start from a generic PEPS tensor description $\{A_i ^s\}$ and:
\begin{itemize}
\item Linearize the computation of its stabilizer purity as shown in Eq.~\eqref{eq:sp-normalized} at the cost of now working with four copies of the initial state $\unpeps$;
\item Notice that the obtained state is still a PEPS whose bond dimension has increased only polynomially with respect to the initial one;
\item Incorporate the projector $Q_1$ at each site obtaining a new PEPS $\ket*{\widetilde{\Phi}(A)}$ living in the tensor product of four copies of the initial physical Hilbert space. This PEPS has the property that
\begin{equation}
    \frac{\braket*{\widetilde{\Phi}(A)}}{\braket*{\widetilde{\psi}(A)}^4}=\stp_2(\ket{\psi(A)})\,.
\end{equation}
\item Using a single call to a specific $\sharpP$ oracle and polynomial post-processing  we are able to extract numerator and denominator of this ratio and then obtain the stabilizer purity.
\end{itemize}
The only previously known hardness result needed from this is $\sharpP-$completeness of generic PEPS contractions, shown in \cref{thm:peps-contraction-hardness} \cite{Schuch_Wolf_Verstraete_Cirac_2007}.

The proof of $\sharpP-$hardness is more involved than the previous one: due to the presence of the $Q_1$ projector at every physical site, it is not immediate whether a generic PEPS contraction could be written as a PEPS stabilizer purity computation. Thus, we were unable to directly use the PEPS contraction hardness results of \cite{Schuch_Wolf_Verstraete_Cirac_2007} to prove $\sharpP-$hardness of our problem. We then follow a different route and directly encode a $\sharpP$ function in the amplitude of a specifically built PEPS, recovering its value by computing the stabilizer purity. The proof follows the following steps:
\begin{itemize}
    \item Starting from a $\sharpP-$problem, namely counting the satisfying assignments of a Boolean formula, we show how such a problem can be encoded in a scalar TN contraction of polynomial size and rational entries (\cref{lem:boolean-to-tn} below, from Ref.~\cite{Kourtis_Chamon_Mucciolo_Ruckenstein_2019})
    \item Given the scalar TN contraction encoding the $\sharpP$ function, we then build a specific PEPS with bond dimension $D=4$ from the original tensors, such that the value of the desired contraction is encoded in a specific single marked qubit amplitude of said PEPS, leaving all other qubits in the state $\ket{0}$. In particular, the state of the marked qubit reads
    \begin{equation}
           \ket{\omega_t}=\frac{t\ket{0}+\ket{1}}{\sqrt{1+t^2}}
    \end{equation}
    where the amplitude $t$ encodes the value of the scalar contraction.
    The building of this PEPS involves basic direct sums of the initial tensors, and is more precisely explained in \cref{lem:branch-amplitude}. 
    \item Compute the stabilizer purity of this specifically built PEPS and prove that for states of such a form, the stabilizer purity is an injective function of the marked qubit amplitude $t$ (and thus of the scalar TN contraction hence the $\sharpP$ computation): this is shown in \cref{lem:sp_t}.
    \item As a final step, find the associated $\sharpP$ function value using a binary search algorithm, which runs in polynomial time.
\end{itemize}
As a result, one is able to recover the value of any $\sharpP$ computation from the stabilizer purity of a suitable PEPS. 
We now state the auxiliary lemmas, along with their proofs, needed for a rigorous proof of \cref{thm:sp-hard} below.
\begin{lemma}[Boolean counting as tensor-network contraction \cite{Kourtis_Chamon_Mucciolo_Ruckenstein_2019}]
\label{lem:boolean-to-tn}
Let $C_x(y)$ be a polynomial-size Boolean circuit depending on an input $x$ and a nondeterministic string $y\in\{0,1\}^{m(n)}$.  One can construct in polynomial time a rational-entry scalar TN $T_x$ of polynomial size such that
\begin{equation}
        \cont(T_x)=\#\{y:C_x(y)=1\}.
\end{equation}
More generally, for every $g\in\GapP$, one can construct a signed rational-entry TN $T_{g,x}$ of polynomial size such that
\begin{equation}
        \cont(T_{g,x})=g(x).
\end{equation}
\end{lemma}

\begin{proof}
The construction is the standard circuit-to-tensor-network encoding \cite{Kourtis_Chamon_Mucciolo_Ruckenstein_2019}. It essentially amounts to reinterpreting the circuit diagram as a scalar tensor-network contraction. Each Boolean wire carries a two-dimensional index.  {Delta} tensors enforce equality of repeated wire values, and gate tensors enforce the truth table of the corresponding Boolean gate.  Summing over all internal wire indices performs the sum over assignments.  Fixing the output index to $1$ makes the contraction count accepting assignments.  If $g=f_+-f_-$, construct networks for $f_+$ and $f_-$ and take their direct sum, multiplying one tensor of the second branch by $-1$.  This gives a signed TN with contraction $g(x)$.  See also the tensor-network counting constructions of Refs.~\cite{Schuch_Wolf_Verstraete_Cirac_2007,Kourtis_Chamon_Mucciolo_Ruckenstein_2019}. 
\end{proof}

At this point, it is important to notice that the TN obtained from a generic Boolean circuit lives on a generic graph with bounded connectivity rather than a square lattice where we build the PEPS in this work. However, if a construction is first described on a general connected graph, it can be embedded into a square-lattice with only polynomial overhead in the number of physical sites by the standard padding and embedding used in PEPS contraction hardness constructions. This involves routing virtual indices along lattice paths using identity tensors, and implementing crossings and permutations by fixed local swap tensors. Consequently, every scalar contraction considered in the reductions can equivalently be realized as the contraction of a square-lattice PEPS without affecting polynomial-time constructibility and thus we can consider all the scalar tensor contractions arising from the Boolean circuits as TN living on a square lattice with no loss of generality.

The following lemma then shows how to encode the contraction of a scalar TN in the one-qubit amplitude of a PEPS.

\begin{lemma}[Branch-amplitude PEPS]
\label{lem:branch-amplitude}
Let $\{T^{(h)}\}$ be a finite set of scalar TN on a square or rectangular lattice $\Lambda$, each with bond dimension $D^{(h)}$, local tensors $B_v^{(h)}$, and contraction $\cont(T^{(h)})$. Then one can construct in polynomial time a PEPS on $\Lambda$ with bond dimension $D=\sum_h D^{(h)}$ and physical dimension $d=2$ whose unnormalized state is
\begin{equation}
        \sum_{h}
        \cont(T^{(h)})\ket*{s_h}_{v_\star}\ket*0^{\otimes(|\Lambda|-1)}.
        \label{eq:branch-state}
\end{equation}
where $s_h$ is an arbitrarily chosen computational basis state of a fixed vertex $v_\star$.

\end{lemma}

\begin{proof}
We provide an explicit construction.  The PEPS is placed on the same lattice $\Lambda$ and at every unmarked site $v\neq v_\star$, we put the direct sum of the tensors $B_v ^{(h)}$ in the physical component $\ket*0$ and put zero in the physical component $\ket*1$:
\begin{equation}
        P_v^0:=\bigoplus_{h} B_v ^{(h)},
        \qquad
        P_v^1:=0.
\end{equation}
At the marked site $v_\star$, we build the PEPS tensors in a similar way, but place the tensors $B_{v_\star}^{(h)}$ in the physical component $\ket*{s_h}$ where we want to encode their contraction and zero in the other components:
\begin{equation}
        P_{v_\star}^{0}:=\bigoplus_{h}B_{v_\star}^{(h)}\delta_{s_h,0} \,,
        \qquad
        P_{v_\star}^{1}:=\bigoplus_{h}B_{v_\star}^{(h)}\delta_{s_h,1}\,.
\end{equation}
Due to the direct-sum structure of the virtual space, each branch $h$ contributes $\cont(T^{(h)})\ket*{s_h}_{v_\star}\ket*0^{\otimes(|\Lambda|-1)}$.
Summing over the possible global branches gives exactly \eqref{eq:branch-state}.
\end{proof}

As a final step, the following lemma shows that the stabilizer purity of the specific PEPS built from the previous lemma is an injective function of the contraction numerical value when restricted to certain intervals, so that one is able to reconstruct the value of the contraction via binary search \cite{Knuth1973}. 

\begin{lemma}\label{lem:sp_t}
Given a state of the form $\ket*{\omega_t}\otimes \ket{0}^{\otimes M}$, with 
\begin{equation}
        \ket{\omega_t}=\frac{t\ket{0}+\ket{1}}{\sqrt{1+t^2}}
        \label{eq:omega-t-def}
\end{equation}
and $t\in \mathbb{R}$, its $2-$stabilizer purity reads 
\begin{equation}
    \stp_2(\omega_t)=p_2 (t):=\frac{t^8+14t^4+1}{(1+t^2)^4}.
    \label{eq:r-t}
\end{equation}
The function $p_2 (t)$ is strictly decreasing on $0\le t\le \sqrt2-1$ and strictly increasing for $t\geq \sqrt{2}+1$.
\end{lemma}
\begin{proof}
        Equation~\eqref{eq:r-t} follows by inserting $\ket*{\omega_t}$ into \eqref{eq:stp_def}.  Writing $u=t^2$, one has
\begin{equation}
        p_2 (t)=\frac{1+14u^2+u^4}{(1+u)^4}
\end{equation}
and
\begin{equation}
        \frac{d r}{d u}
        =
        \frac{4(u-1)(u^2-6u+1)}{(1+u)^5}.
\end{equation}
For $0\le u\le(\sqrt2-1)^2$, the factor $u-1$ is negative while $u^2-6u+1$ is nonnegative, with equality only at the endpoint.  Thus $r$ is strictly decreasing on $0\le t\le\sqrt2-1$. Analogously, for $u\geq (\sqrt{2}+1)^2$, both factors are nonnegative hence $r$ is strictly increasing for $t\geq\sqrt{2}+1$.
\end{proof}
We are now ready to prove \cref{thm:sp-hard}.
\begin{proof}[Proof of \cref{thm:sp-hard}]
We first prove the upper bound, namely that the stabilizer purity of a PEPS $\ket*\psi$ can be obtained by having a single call to a $\sharpP$ oracle and efficient postprocessing.
Computing $2-$stabilizer purity of the resulting normalized PEPS $\ket*{\psi(A)}$ amounts to 
\begin{equation}
    \stp_2(\ket*{\psi(A)})=2^N \frac{\bra*{\widetilde{\psi}(A)}^{\otimes 4}{Q_N}\ket*{\widetilde{\psi}(A)}^{\otimes 4}}{\braket*{\widetilde{\psi}(A)}^4}=2^N\frac{S(A)}{Z(A)^4}\,.
\end{equation}
Recalling the factorization property of $Q_N$ from \cref{eq:Q-factorization} and the fact that it is an orthogonal projector, together with the fact that if $\ket*{\widetilde{\psi}(A)}$ is an unnormalized PEPS with tensor entries $A_i ^{s_i}$, then, also $\ket*{\widetilde{\psi}(A)}^{\otimes 4}$ is a PEPS with tensor $A^{\ot 4}=\raisebox{-0.33\height}{\includegraphics[height=4 em]{figures/replicated_tensor.pdf}}=\raisebox{-0.38\height}{\includegraphics[height=3.3 em]{figures/thick_tensor.pdf}}$, we now build the PEPS $\ket*{\widetilde{\Phi}(A)}$ by defining its tensors as $Q_1A^{\ot 4}=\raisebox{-0.38\height}{\includegraphics[height=3 em]{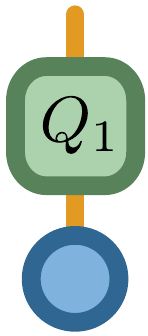}}
$. At this point it is straightforward to see that
\begin{equation}
     \braket*{\widetilde{\Phi}(A)}=\bra*{\widetilde{\psi}(A)}^{\otimes 4}{Q_N}\ket*{\widetilde{\psi}(A)}^{\otimes 4}
\end{equation}
hence
\begin{equation}\label{spratio}
     2^N \frac{\braket*{\widetilde{\Phi}(A)}}{\braket*{\widetilde{\psi}(A)}^4}=\stp_2(\ket*{\psi(A)})\,.
\end{equation}
Notice that this construction is the 2D analogue of the one seen in \cite{Haug_Piroli_2023} used to efficiently compute SE for Matrix Product States.
We know from \cref{thm:peps-contraction-hardness}, shown in Ref.~\cite{Schuch_Wolf_Verstraete_Cirac_2007}, that computing PEPS norms is $\sharpP-$complete, hence they can be computed by having access to a $\sharpP$ oracle. From Eq. \eqref{spratio}, it may seem that one needs \textit{two} calls to a $\sharpP$ oracle to determine PEPS stabilizer purity, since we need the norm of both $\ket*{\widetilde{\Phi}}$ and $\ket*{\widetilde\psi}$. However, we now show that both quantities can be accessed from a single call to a different $\sharpP$ oracle. We denote the two $\sharpP$ functions of the norms of the two PEPS as
\begin{equation}
    \begin{split}
        \braket*{\widetilde{\Phi}(A)}&:=f_1\\
        \braket*{\widetilde{\psi}(A)}&:=f_2
    \end{split}
\end{equation}
Then, choose a known polynomial $h(n)$ strictly exceeding the witness-length
bounds of these two machines, so $0\le f_j<2^{h(n)}$. Define then the following $\sharpP$ function (recall that $\sharpP$ is closed under sum and integer multiplication~\cite{FortnowRogers1999ComplexityLimitations}):
\begin{equation}
    C(A)=\sum_{j=1}^2 2^{(j-1)h(n)}f_j(A).
\end{equation}
Finally, its base-$2^h$ digits recover both
counts:
\begin{equation}
f_j(A)=\left\lfloor\frac{C(A)}{2^{(j-1)h(n)}}\right\rfloor\bmod2^{h(n)}.
\end{equation}
Hence computing the ratio of the two quantities and multiplying by $2^N$ one obtains $\stp_2(\ket*{\psi(A)})$ using a single call to a $\sharpP$ oracle.

We now show $\sharpP$-hardness.  Let $f\in\sharpP$.  For input $x$, choose a Boolean circuit whose accepting assignments are counted by $f(x)$, and use \cref{lem:boolean-to-tn} to construct a scalar TN $T_f(x)$ with
\begin{equation}
        \cont(T_f(x))=f(x).
\end{equation}
Let $p(n)$ be a polynomial upper bound on the length of the witness of the function $f$, so $0\le f(x)\le2^{p(n)}$.  Let $M+1$ be the number of physical sites needed to represent the scalar TN $T_f(x)$ on a square lattice and choose $ \lambda:=\frac{242}{100}>\sqrt{2}+1$ (the reason for this will become clear subsequently) and use \cref{lem:branch-amplitude} with three branches.  The first branch is $T_f(x)$, and is routed to the marked physical value $0$.  The second branch is a scalar-one network on the same graph, routed to the marked physical value $1$, and the third branch is a scalar-one network scaled to have contraction $\lambda$, routed to the marked physical state 0. The resulting unnormalized PEPS, with bond dimension 4, is
\begin{equation}
        \ket*{\widetilde\phi_f}=[(\lambda+f(x))\ket*{0}_{v_\star}+\ket{1}_{v_\star}]\ket*0^{\otimes M}.
\end{equation}
After normalization,
\begin{equation}
        \ket*{\phi_f}
        =
        \ket*{\omega_t}\otimes\ket*0^{\otimes M},
        \quad
        t:=\lambda + f(x).
\end{equation}
Tensoring with the stabilizer product state $\ket*0^{\otimes M}$ leaves stabilizer purity unchanged, so
\begin{equation}
        \stp_2(\phi_f)
        =
        p_2 (t)
        =
        \frac{t^8+14t^4+1}{(1+t^2)^4}.
\end{equation}
It remains to justify that the value of $f(x)$ can be recovered from the
exact rational number $p_2(\lambda + f(x))$ in polynomial time.  The possible values of $f(x)$ form the integer interval $\{0,1,\ldots,2^{p(n)}\}$.
Although this set has exponentially many elements, its endpoints have
$O(p(n))$ bits, and binary search needs only $O(p(n))$ comparisons.
Concretely, initialize two integers $\ell_0:=0$ and $\ell_{\max}:=2^{p(n)}$.

At each step set $ k:=\floor{\dfrac{\ell_0+\ell_{\max}}{2}}$
and compare the candidate value $p_2(\lambda+k)$ with the oracle value
$p_2(\lambda+f(x))$.
Since $p_2$ is strictly increasing (thus injective) on the interval $t\geq \lambda$ by \cref{lem:sp_t}, this
comparison has the following meaning:
\begin{equation}
        p_2(\lambda+k)=p_2(\lambda+f(x))\iff k=f(x),
\end{equation}
while
\begin{equation}
        p_2(\lambda+k)>p_2(\lambda+f(x))\iff k>f(x),
\end{equation}
and
\begin{equation}
        p_2(\lambda+k)<p_2(\lambda+f(x))\iff k<f(x).
\end{equation}
Thus, if $p_2(\lambda+k)<p_2(\lambda+f(x))$, one replaces the lower endpoint by
$\ell_0:=k+1$; if $p_2(\lambda+k)>p_2(\lambda+f(x))$, one replaces the upper endpoint by
$\ell_{\max}:=k-1$.  Each step halves the number of remaining candidate
integers, so after at most
\begin{equation}
       \ceil{ \log (2^{p(n)}+1)}
        =
        O(p(n))
\end{equation}
steps the unique integer $k=f(x)$ is found.
Every comparison is polynomial-time exact rational arithmetic, and
the full reconstruction of $f(x)$ is a polynomial-time exact postprocessing
procedure. Given the exact value $\stp_2(\phi_f)$, one recovers the integer $f(x)$ by exact binary search over $k\in\{0,1,\ldots,2^{p(n)}\}$, comparing the rational values $p_2(\lambda+k)$ with the oracle output.  The search uses $O(p(n))$ exact rational comparisons and can therefore be computed in polynomial time.  Thus exact stabilizer-purity evaluation computes an arbitrary $\sharpP$ function under exact postprocessing.
\end{proof}

We now prove that evaluating $\alpha-$stabilizer purity is $\sharpP$ hard for every integer value $\alpha>1$:
\begin{theorem}\label{thm:sp_alpha_hard}
    Given the tensor description $\{A_i\}$ of a PEPS on a square lattice, computing $\stp_\alpha(\ket*{\psi(A)})$ is $\sharpP-$hard for every {constant integer} $\alpha>1$ and can be evaluated with access to a $\sharpP$ oracle and polynomial-time post-processing.
\end{theorem}
\begin{proof}
    The hardness proof is almost identical to that of $\sharpP-$hardness of the 2-stabilizer purity:  we follow the same procedure as in the proof of Theorem \ref{thm:sp-hard} by encoding the value of a $\sharpP$ function $f$ in the single qubit amplitude of a PEPS and then computing its $\alpha-$stabilizer purity, which yields
    \begin{equation}
        \stp_\alpha(\phi_f):=p_\alpha(t)=\frac12\left[1+\left(\frac{2t}{1+t^2}\right)^{2\alpha}+\left(\frac{1-t^2}{1+t^2}\right)^{2\alpha}\right]\,.
    \end{equation}
    This function is strictly monotone (thus injective) for $t\geq \sqrt{2}+1$: to see that, we compute its derivative. This reads
    \begin{equation}
       \stp_\alpha'(t)=\dv{p_\alpha(t)}{t}=\frac{4 \alpha  t^2 \left(\frac{\left(t^2-1\right)^2}{\left(t^2+1\right)^2}\right)^{\alpha }-4^{\alpha } \alpha  \left(t^2-1\right)^2 \left(\frac{t^2}{\left(t^2+1\right)^2}\right)^{\alpha }}{t \left(t^4-1\right)}
    \end{equation}
    For $t>\sqrt{2}+1$, the expression can be written as
    \begin{equation}
        \stp_\alpha'(t)=
        \frac{4\alpha t(t^2-1)}{(t^2+1)^{2\alpha+1}}
        \left[(t^2-1)^{2\alpha-2}-(2t)^{2\alpha-2}\right].
    \end{equation}
    Since $\alpha>1$, $t>0$, and $t^2-1>0$, the prefactor is strictly positive. Moreover,
    \begin{equation}
        t>\sqrt{2}+1
        \quad\Longrightarrow\quad
        t^2-1>2t,
    \end{equation}
    and, because $2\alpha-2>0$,
    \begin{equation}
        (t^2-1)^{2\alpha-2}>(2t)^{2\alpha-2}.
    \end{equation}
    Hence the quantity in brackets is also strictly positive, proving that
    \begin{equation}
        \stp_\alpha'(t)>0,
        \qquad
        t>\sqrt{2}+1.
    \end{equation}
    Since the $\stp_\alpha(\phi_f)$ is injective, we can again use a binary search algorithm to recover the value of $t$ and thus the value of $f(x)$.

    For containment, we build a $2\alpha-$replica TN to compute
    \begin{equation}
U_\alpha=\Tr\!\left[
\left(\sum_{P\in\pauli_N}P^{\otimes2\alpha}\right)
(\ketbra*{\widetilde{\psi}(A)})^{\otimes2\alpha}\right].
\end{equation}
This uses $2\alpha$ ket copies and $2\alpha$ bra copies. Replicating
the ket PEPS gives bond dimension $D^{2\alpha}$; including the bra
gives a scalar double layer with dimension $D^{4\alpha}$. The
physical/operator factors are constants for fixed $\alpha$. Notice that the operator $\sum_{P\in\pauli_N}P^{\otimes2\alpha}$ is proportional to a projector for even $\alpha$, whereas it is proportional to a unitary operator for $\alpha$ odd~\cite{Bittel2026operational}. In both cases, the operator can be incorporated into the replicated tensor network. The $\alpha-$stabilizer purity is then $\frac{U_\alpha}{2^N Z(A)^{2\alpha}}$, which can be computed using the $\sharpP$ oracle procedure described in the proof of Theorem~\ref{thm:sp-hard}.
\end{proof}

\subsection{\texorpdfstring{$\alpha-$}{}stabilizer entropy}

We are now ready to state the main result, namely the $\sharpP$ hardness of evaluating $\alpha-$SE.

\begin{theorem}[Complexity of the $\alpha$-SE]
\label{thm:complexity_stabilizer_entropy}
Fix a constant integer $\alpha>1$, and consider rational-entry $N-$qubit PEPS on a square lattice with a description of length $n=\poly(N)$, with the promise that $Z(A)>0$.
Problem~\ref{prob:dyadic_stabilizer_entropy} with $m=\Omega(\poly(n)+\log (\alpha))$ is
$\sharpP$-hard under polynomial-time Turing reductions. Moreover this problem can be solved by an algorithm in $\FP^{\sharpP}$.

More precisely, exact evaluation of $\stp_\alpha$ and dyadic
evaluation of $M_\alpha$ up to $m=\Omega(\poly(n)+\log (\alpha))-$digits reduce to one another using one oracle
query and polynomial-time postprocessing. The reduction succeeds for every oracle answer
satisfying the prescribed error bound.
\end{theorem}

The theorem relates two ways of describing the same quantity: the stabilizer purity as an exact rational number, and the SE as a numerical approximation with a prescribed accuracy. Although these quantities are connected by a logarithm, their computational equivalence requires some care. The stabilizer purity has a finite exact representation through its numerator and denominator, whereas the entropy is generally irrational. The central observation is that sufficiently many digits of the entropy nevertheless determine stabilizer purity exactly, because the possible rational values obey a denominator bound.

The first direction is straightforward: once the exact stabilizer purity is available, evaluating its logarithm gives the entropy. The numerical calculation can be carried out efficiently by rescaling the logarithm's argument to a bounded interval and
evaluating a convergent series with a controlled truncation error. A few additional digits of internal precision account for the final rounding to the required dyadic grid.
Consequently, an exact-purity oracle allows us to produce a valid entropy approximation with only polynomial-time additional computation.

The converse direction explains why arbitrary-precision entropy evaluation inherits the hardness of exact purity evaluation. Writing
\begin{equation}
    M=M_\alpha(\psi), \qquad p:=\stp_\alpha(\psi)=2^{-(\alpha-1)M},
\end{equation}
the argument proceeds in three steps:
\begin{itemize}
    \item \textit{Convert an entropy estimate into a controlled purity estimate.} Given an approximation $\widetilde M$ to $M$, evaluate $2^{-(\alpha-1)\widetilde M}$ numerically. There are two sources of error: the uncertainty in the supplied entropy and the numerical error in evaluating the exponential. The mean-value theorem controls the first contribution: on the nonnegative entropy range, the magnitude of the derivative of the inverse transformation is bounded by $(\alpha-1)\ln2$. Thus, for fixed $\alpha$, an entropy error of order $2^{-m}$ produces a purity error of the
    same order. The second contribution can be made equally small by choosing sufficient numerical precision.
    \item \textit{Make the uncertainty smaller than the separation
    between candidate fractions.} Given that stabilizer purity is a rational number with a known denominator bound, then it suffices to estimate the SE up to an error that is smaller than the smallest possible separation between two neighboring fraction with the same denominator. This is what yields $m=\Omega(\poly(n)+\log  \alpha)$.
    \item \textit{Recover that fraction efficiently.} Once the approximation uniquely identifies the fraction, standard numerical methods recover its exact value in polynomial time.~\cite{BrentZimmermann2010}
    
\end{itemize}

Although exact recovery requires an exponentially small error,
only polynomially many digits of the entropy are needed.
A single entropy query at this precision therefore allows us
to recover the exact stabilizer purity with polynomial-time
postprocessing. Combined with the reverse conversion from
purity to entropy, this establishes the claimed computational
equivalence.

\begin{proof}
Let $n$ denote the length of the explicit binary encoding of the
PEPS tensors, and write
\begin{equation}
    p_\alpha(A)
    \coloneqq \stp_\alpha\bigl(\ket*{\psi(A)}\bigr),
    \qquad
    M \coloneqq M_\alpha\bigl(\ket*{\psi(A)}\bigr).
\end{equation}
We use the normalization
\begin{equation}
    p_\alpha(A)
    =
    2^{N(\alpha-1)}
    \sum_{P\in\pauli_N}\Xi_P(\psi(A))^\alpha,
    \qquad
    M=\frac{1}{1-\alpha}\log  p_\alpha(A).
\end{equation}
Consequently,
\begin{equation}
    p_\alpha(A)=2^{(1-\alpha)M}.
    \label{eq:purity_from_alpha_entropy}
\end{equation}
Since the characteristic distribution of a normalized pure state
satisfies
$0\le\Xi_P\le2^{-N}$ and $\sum_P\Xi_P=1$, we have
$p_\alpha(A)\le1$.
The identity Pauli contributes $2^{-N}$ to the expression
$p_\alpha(A)=2^{-N}\sum_P
|\langle\psi(A)|P|\psi(A)\rangle|^{2\alpha}$.
Thus
\begin{equation}
    2^{-N}\le p_\alpha(A)\le1,
    \qquad
    0\le M\le\frac{N}{\alpha-1}.
    \label{eq:dyadic-entropy-range}
\end{equation}
\medskip
\noindent

One can find a polynomial $q$ such that the reduced denominator of
$p_\alpha(A)$ is at most $B=2^{q(n)}$.
Moreover, the preceding exact-purity evaluation result gives
$p_\alpha(A)$ using a $\sharpP$ oracle and polynomial-time
postprocessing.
Given the exact rational $p_\alpha(A)$ and the requested digit precision $m$
, compute a rational number $\widehat M$ satisfying
\begin{equation}
    |\widehat M-M|\le 2^{-m-2}.
\end{equation}
Certified evaluation of the logarithm achieves this accuracy in
time polynomial in $n+m$.
The logarithm and exponential evaluations used below can be
performed with certified accuracy by standard argument reduction
and power-series methods, with suitable control of working
precision; see \cite[Secs.~4.3--4.4]{BrentZimmermann2010}.
For the logarithm, we use
\begin{equation}
  \ln y
  = 2\sum_{j=0}^{\infty}
      \frac{1}{2j+1}
      \left(\frac{y-1}{y+1}\right)^{2j+1},
\end{equation}
which follows from the logarithmic series given in
\cite[Sec.~4.4.2]{BrentZimmermann2010}.
Once the approximation satisfies the error bound established
above, the exact rational purity can be recovered in polynomial
time by continued fractions, using the previously established
denominator bound; see the rational-recovery step in
\cite[Sec.~5]{Shor1997}.
Together with a certified approximation of $\ln2$, this gives
the required polynomial-time evaluation of $\log p_\alpha(A)$.
Return
\begin{equation}
    k
    =
    \max\left\{
        0,\,
        \left\lfloor 2^m\widehat M+\frac12\right\rfloor
    \right\}.
\end{equation}
Rounding to the dyadic grid contributes an error of at most
$2^{-m-1}$.
Since $M\ge0$, replacing a negative rounded value by zero cannot
increase its error.
Hence
\begin{equation}
    \left|\frac{k}{2^m}-M\right|
    \le 2^{-m-2}+2^{-m-1}
    =3\cdot2^{-m-2}
    <2^{-m}.
\end{equation}
The output has $O(m+\log (N+1)+1)$ bits by
Eq.~\eqref{eq:dyadic-entropy-range}.
Thus an $\FP^{\sharpP}$ algorithm produces a valid answer
to the dyadic evaluation problem.
The same procedure also gives a reduction using one query to an
exact-purity oracle.

Now we recover the exact stabilizer purity from a dyadic entropy oracle.
Fix the PEPS description $A$, and abbreviate
$q=q(n)$ and $B=B(A)$.
Two distinct reduced fractions whose positive denominators are
at most $B$ obey
\begin{equation}
    \left|\frac ab-\frac{a'}{b'}\right|
    =
    \frac{|ab'-a'b|}{bb'}
    \ge\frac1{B^2}.
    \label{eq:alpha_purity_separation}
\end{equation}
We will approximate $p_\alpha(A)$ accurately enough to identify
its reduced fraction uniquely.
Query the dyadic entropy oracle at precision
\begin{equation}
    m_\star
    \coloneqq
    2q+\left\lceil\log \alpha\right\rceil+2.
\end{equation}
Because $q$ is polynomially bounded in $n$, the output number has polynomial length.
Let
\begin{equation}
    \widetilde M=\frac{k}{2^{m_\star}}
\end{equation}
be any valid oracle answer.
By the output requirements,
\begin{equation}
    \widetilde M\ge0,
    \qquad
    |\widetilde M-M|\le2^{-m_\star}.
\end{equation}
Define
\begin{equation}
    f_\alpha(x)\coloneqq 2^{(1-\alpha)x}.
\end{equation}
Its derivative is
\begin{equation}
    f_\alpha'(x)
    =
    -(\alpha-1)\ln2\,2^{-(\alpha-1)x}.
\end{equation}
If $\widetilde M=M$, the following error bound is immediate.
Otherwise, the mean-value theorem gives a point $\xi$ between
$M$ and $\widetilde M$ such that
\begin{equation}
    f_\alpha(\widetilde M)-f_\alpha(M)
    =
    f_\alpha'(\xi)(\widetilde M-M).
\end{equation}
Both endpoints are nonnegative, so $\xi\ge0$ and
$2^{-(\alpha-1)\xi}\le1$.
Taking absolute values therefore gives
\begin{align}
    \left|f_\alpha(\widetilde M)-p_\alpha(A)\right|
    &\le
    (\alpha-1)\ln2\,|\widetilde M-M|
    \notag\\
    &\le
    (\alpha-1)\ln2\,2^{-m_\star}.
    \label{eq:alpha_purity_error}
\end{align}
Thus the inverse transformation has a Lipschitz bound on the
nonnegative half-line that is independent of the PEPS size.
The number $f_\alpha(\widetilde M)$ need not be rational.
Using certified arbitrary-precision exponentiation, compute a
dyadic rational $\widetilde p$ with
\begin{equation}
    \left|\widetilde p-f_\alpha(\widetilde M)\right|
    \le2^{-m_\star}.
\end{equation}
This takes polynomial time: the input has polynomial bit length,
the requested precision is polynomial, and
$0\le\widetilde M\le N/(\alpha-1)+1$.
Argument reduction and power-series evaluation with certified
remainder bounds suffice.
Combining the two errors yields
\begin{align}
    |\widetilde p-p_\alpha(A)|
    &\le
    \left(1+(\alpha-1)\ln2\right)2^{-m_\star}
    \notag\\
    &<
    \alpha\,2^{-m_\star}
    \le\frac1{4B^2}.
    \label{eq:total_alpha_purity_error}
\end{align}
Here we used $\ln2<1$ and the definition of $m_\star$.
If necessary, clamp $\widetilde p$ to $[0,1]$; since
$p_\alpha(A)\in[0,1]$, this does not increase the error.
Equation~\eqref{eq:alpha_purity_separation} shows that at most one
reduced fraction with denominator at most $B$ can satisfy
\begin{equation}
    \left|\widetilde p-\frac ab\right|<\frac1{4B^2}.
\end{equation}
Such a fraction exists, namely $p_\alpha(A)$.
It can be found in polynomial time using continued fractions.
Indeed, the continued-fraction criterion states that a reduced
fraction satisfying
\begin{equation}
    \left|\widetilde p-\frac ab\right|<\frac1{2b^2}
\end{equation}
is a convergent of $\widetilde p$.
The true fraction satisfies this inequality because $b\le B$ and
Eq.~\eqref{eq:total_alpha_purity_error} holds.
Compute the finite continued-fraction expansion of the rational
number $\widetilde p$, enumerate its convergents, and retain the
unique one with denominator at most $B$ and error less than
$1/(4B^2)$.
All comparisons are exact rational comparisons, and the Euclidean
algorithm~\cite{Burton_2010} and convergent computation take polynomial time in the
relevant bit lengths.

We have therefore recovered $p_\alpha(A)$ exactly using one
dyadic-entropy query and polynomial-time postprocessing.
The argument used only the defining guarantees on the oracle
answer, so it succeeds for every valid output of the approximation
relation.
Composing this reduction with the hardness reduction for exact
$\stp_\alpha$ evaluation in
Theorem~\ref{thm:sp_alpha_hard} proves $\sharpP$-hardness under
polynomial-time Turing reductions.
\end{proof}

    Theorem~\ref{thm:complexity_stabilizer_entropy} does not assert
    that $M_\alpha$ is itself a $\sharpP$ function, since
    $M_\alpha$ is generally irrational, whereas functions in
    $\sharpP$ are nonnegative and integer-valued. The hardness result
    concerns arbitrary-precision evaluation and requires polynomially
    many accurate binary digits, corresponding to exponentially small
    additive error. It does not, by itself, establish hardness for
    constant or inverse-polynomial additive accuracy.

\section{Exact PEPS stabilizer membership}
\label{sec:exact-membership}

We now move to the problem of deciding whether a certain PEPS is a stabilizer state given its tensor description. This problem has been approached in a more general setting of density matrices in~\cite{unbearable}, and is known to be computationally efficient in 1D TN states~\cite{Haug_Piroli_2023}, so it is interesting to explore the complexity of this problem in a more structurally rich, yet still physically relevant scenario such as 2D PEPS.

\begin{problem}[PEPS exact stabilizer membership --  $\stabpeps$]\label{prob:exact-membership}
Consider a nonzero PEPS over the $N-$qubit square lattice with bond dimension $D$, specified by tensors $A=\{A_i^{s_i}\}$ with rational entries.
	Let $\ket*{\widetilde{\psi}(A)}$ be the unnormalized PEPS generated by $A$.
	The YES instances of the problem are those for which the normalized state
	\begin{equation}
		\ket*{\psi(A)}
		:=
		\frac{\ket*{\widetilde\psi(A)}}{\sqrt{\braket*{\widetilde\psi(A)}}}
	\end{equation}
	is a pure stabilizer state, with every other instance being a NO instance.
	Equivalently,
	\begin{equation}
		A\in \stabpeps
		\iff
		\ket*{\psi(A)}\in \STAB_N .
	\end{equation}
\end{problem}

\begin{theorem}[Exact PEPS stabilizer membership is $\CeqP$-complete]
\label{thm:exact-membership}
Problem~\ref{prob:exact-membership} with bond dimension $D= 33$ is $\CeqP$-complete under polynomial-time many-one reductions.
\end{theorem}

The reduction constructed in the above Theorem maps an arbitrary zero test $g(x)=0$, with $g\in\GapP$, to a rational-entry, constant-bond-dimension PEPS whose normalized state is a stabilizer state if and only if the zero test accepts.  Given the collapse of complexity classes established in \cref{subsec:ceqp-significance}, it is highly unlikely that either a classical or quantum computer can thus decide the exact stabilizer membership of PEPS with polynomial resources.

As in~\cref{thm:sp-hard}, the proof is split into two parts: containment and hardness, which together prove completeness. The proof of containment relies on the previous result of $\sharpP-$hardness of $\stp_2$ and generic PEPS contractions proven in~\cref{thm:sp-hard} and~\cite{Schuch_Wolf_Verstraete_Cirac_2007}:
\begin{itemize}
    \item By using the stabilizer purity criterion, namely $\stp_2(\psi)=1\iff \psi\in\STAB$ shown in~\cref{lem:sp-faithfulness}~\cite{Leone_Oliviero_Hamma_2022}, we write the stabilizer membership test for PEPS as the difference between the stabilizer purity and the fourth power of the PEPS norm. \item Since $\sharpP$ is closed under multiplication~\cite{valiant_permanent}, we recognize that we have written the stabilizer membership test as the difference of two $\sharpP$ functions which is zero if and only if the considered PEPS is a stabilizer state: by definition, this is a zero-test for a $\GapP$ function, hence a $\CeqP$ problem.
\end{itemize}

For the proof of hardness, we follow a similar path to the $\sharpP$ proof of hardness of computing stabilizer purity. The reduction chain is summarized in \cref{fig:reduction-chain}, and the steps are explained below:
\begin{itemize}
    \item We encode the value of a $\GapP$ function in a scalar TN contraction using~\cref{lem:boolean-to-tn} again;
    \item We then encode the value of the contraction in the single-qubit amplitude of a specifically built PEPS, while leaving the other qubits in the $\ket 0$ state. This time, though, we do not encode it directly but rather in the argument of a rational function $F$
    \begin{equation}
        F(k):=\frac{c k^2}{1+k^2},
        \qquad 0<c<\sqrt2-1,
        \label{eq:F-def}
\end{equation}
where $c$ is a fixed rational constant. Its key properties are
\begin{equation}
        F(0)=0,
        \qquad
        k\in\mathbb Z\setminus\{0\}\Rightarrow \frac c2\le F(k)<c.
        \label{eq:F-properties}
\end{equation}
\item The single qubit state that encodes the value of $g(x)$ then reads 
\begin{equation}
\label{phig}
    \ket*{\phi_g}=
        \frac{ F(g(x))\ket*0
        +
        \ket*1}{\sqrt{1+F(g(x))^2}}.
\end{equation}
The complete construction of this encoding is shown in~\cref{lem:saturated-gadget}. This rational function construction is particularly useful for this problem, as $F$ has the ability of distinguishing instances of $g(x)=0$ from $g(x)\neq 0$ as seen in~\cref{eq:F-properties}, and excluding the other values out of our scope. In fact, a state of the form of~\cref{phig} can be a stabilizer state if and only if the $\ket0$ amplitude is either $0$ or $\pm 1$. 
\item Since, as mentioned before $-1<F(x)<1$, $\ket{\phi_g}$ being a stabilizer unequivocally implies that $g(x)=0$, hence proving the reduction from $\CeqP$ problems to instances of exact PEPS stabilizer membership.
\end{itemize}

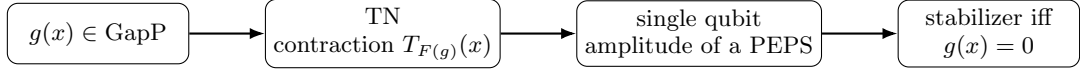
\begin{figure}[t!]
\centering
\begin{tikzpicture}[node distance=1.2cm and 1cm, every node/.style={font=\small}]
\tikzset{
box/.style={rectangle,draw,rounded corners,align=center,minimum width=2.4cm,minimum height=0.8cm},
arr/.style={-{Latex[length=2mm]},thick}
}
\node[box] (gap) {$g(x)\in\GapP$};
\node[box,right=of gap] (tn) {TN \\ contraction $T_{F(g)}(x)$};
\node[box,right=of tn] (peps) {single qubit \\amplitude of a PEPS};
\node[box,right=of peps] (stab) {stabilizer iff\\$g(x)=0$};
\draw[arr] (gap)--(tn);
\draw[arr] (tn)--(peps);
\draw[arr] (peps)--(stab);
\end{tikzpicture}
\caption{Reduction proving $\CeqP$-hardness of exact PEPS stabilizer membership.  The zero test for a function $g\in \GapP$ is encoded in the amplitude of a single marked qubit of a specifically built PEPS, through the argument of a rational function. The resulting state is a stabilizer state if and only if $g(x)=0$.}
\label{fig:reduction-chain}
\end{figure}

We now show how to encode the value of $F(g(x))$ in the amplitude of a PEPS.

\begin{lemma}
\label{lem:saturated-gadget}

Let $T$ be a polynomial-size rational scalar TN with bond dimension $D$ on a connected graph and let $M+1$ be the number of sites needed for representing $T$ on a square lattice. For every fixed rational $0<c<\sqrt2-1$, one can construct in polynomial time a rational-entry qubit PEPS with bond dimension $D'=2D^2+1$ whose normalized state is
\begin{equation}
        \ket*{\Phi_T}
        =
        \frac{F(\cont(T))\ket*0+\ket*1}{\sqrt{1+F(\cont(T))^2}}
        \otimes\ket*0^{\otimes M},
\end{equation}
with the function $F$ defined in \eqref{eq:F-def}.
\end{lemma}

\begin{proof}    

Let $T':=T\otimes T$, which has contraction $\cont(T)^2$. Construct also a scalar network $T_1$ with contraction value 1 on the same graph.
Multiply one tensor of $T'$ by $c$ and denote the resulting TN as $T_{\mathrm{num}}$, so that $\cont(T_{\mathrm{num}})=c\,\cont(T)^2$.
Then, apply \cref{lem:branch-amplitude} to three branches, routing $T_{\mathrm{num}}$ to $\ket0$ and routing $T_1$ and $T'$ to $\ket1$.  The resulting unnormalized PEPS is
\begin{equation}
\begin{aligned}
        \ket*{\widetilde\Phi_T}
        &=
        c \,\cont(T)^2\ket*0\ket*0^{\otimes M}
        +
        \bigl(1+\cont(T)^2\bigr)\ket*1\ket*0^{\otimes M} \\
        &=
        \bigl(1+\cont(T)^2\bigr)
        \left(
        \frac{c \,\cont(T)^2}{1+\cont(T)^2}\ket*0+
        \ket*1
        \right)\ket*0^{\otimes M}.
\end{aligned}
\end{equation}
Normalizing gives the claimed state.  The bond dimension of the PEPS comes from the sum of the bond dimensions of the TN involved, namely, the bond dimension of $T_{\rm num}$ is $D^2$, as is that of $T'$, and $T_1$ has bond dimension one.
\end{proof}

We are now ready to prove \cref{thm:exact-membership}.

\begin{proof}[Proof of Theorem \ref{thm:exact-membership}]
First we prove containment.  By \cref{lem:sp-faithfulness}, the normalized state is stabilizer if and only if $\stp_2(\psi(A))=1$.  Using \eqref{eq:sp-unnormalized}, this is equivalent to
\begin{equation}
        2^N S(A)-Z(A)^4=0.
        \label{eq:membership-gap-function}
\end{equation}
After clearing denominators, the left-hand side is a GapP function of the PEPS input: $S(A)$ and $Z(A)$ are exact tensor-network contractions, GapP is closed under products and subtraction, and multiplication by the known integer $2^N$ preserves GapP.  Hence the language is in $\CeqP$.

We now prove hardness.  Let $L\in\CeqP$.  Choose $g\in\GapP$ such that
\begin{equation}
        x\in L
        \quad\iff\quad
        g(x)=0.
        \label{eq:ceqp-hardness-start}
\end{equation}
By \cref{lem:saturated-gadget}, construct a PEPS with physical dimension $d=2$ and bond dimension $D=33$ such that its unnormalized state reads
\begin{equation}
        \ket*{\widetilde\phi_g}
        =
         F(g(x))\ket*0\ket*0^{\otimes M}
        +
        \ket*1\ket*0^{\otimes M}.
\end{equation}
The normalized state is then
\begin{equation}
        \ket*{\phi_g}
        =
        \frac{F(g(x))\ket*0+\ket*1}{\sqrt{1+F(g(x))^2}}
        \otimes\ket*0^{\otimes M},
        \qquad
\end{equation}
 The real one-qubit state $(t\ket*0+\ket*1)/\sqrt{1+t^2}$ is a stabilizer state only at $t=0$ or $t=\pm1$ among real $t$ of this form.  The cases $t=\pm1$ are excluded since $F(g(x))=0\iff g(x)=0$ and $-1<c/{2}\leq F(g(x)) < c <1$ for $g(x)\neq 0$.  By tensoring with $\ket*0^{\otimes M}$, we therefore have
\begin{equation}
        \ket*{\phi_g}\in\STAB
        \quad\iff\quad
        F(g(x))=0
        \quad\iff\quad
        g(x)=0
        \quad\iff\quad
        x\in L.
\end{equation}
This gives a polynomial-time many-one reduction from an arbitrary $\CeqP$ language to exact PEPS stabilizer membership.  Together with containment, the problem is $\CeqP$-complete.
\end{proof}

\section{Hardness of deciding local-unitary equivalence to stabilizer states}
\label{subsec:nonlocal-magic}
For a normalized pure state $\ket\psi$ on a qubit bipartition $A|B$ with $N_A+N_B=N$, we say that $\ket\psi$ is \textit{local-equivalent to the stabilizer set} if there exists two unitary operators $U_A$ and $U_B$ acting on $\mathcal{H}_A$ and $\mathcal{H}_B$ respectively, such that $(U_A\otimes U_B)\ket\psi\in\STAB_N$. We denote the set of local-equivalent stabilizers as $\lstab$.

We invoke the following characterization of the
 set of local-equivalent stabilizer set, established in
Ref.~\cite[Lemma~1, Eq.~(44)]{cao2025gravitational}
and Ref.~\cite[Theorem~1 and the discussion following
Eq.~(3)]{sierant2026exactquantificationnonlocalmagic}.

\begin{lemma}[Local-equivalent stabilizer set]
\label{lem:nlm-zero}
For every normalized pure state $\ket\psi$ on a qubit bipartition $A|B$, $\ket\psi\in\lstab$ if and only if the nonzero
Schmidt coefficients of $\ket{\psi}$ are all equal and their
number is $2^k$ for some integer
$0\le k\le\min\{|A|,|B|\}$.
\end{lemma}

We formulate the problem of deciding whether a PEPS description represents a state belonging to $\lstab$.
\begin{problem}[Local-stabilizer equivalence decision]
\label{prob:zero-nlm}
Given a rational-entry tensor
description $T$ of a nonzero $N$-qubit PEPS on a square lattice
and a specified bipartition $A|B$ of its physical qubits,
output YES if corresponding normalized state
$\ket{\psi(T)}$ belongs to $\lstab$, and NO otherwise.
\end{problem}
The next theorem proves the computational hardness of local-stabilizer equivalence decision.
\begin{theorem}[Hardness of $\lstab$ decision]
\label{thm:nlm-hardness}
Problem~\ref{prob:zero-nlm} is $\CeqP$-hard under polynomial-time many-one
reductions, with bond dimension 
$D=33$, and any bipartition $A|B$ of the lattice.
\end{theorem}
\begin{proof}
Let $L\in\CeqP$, and choose $g\in\GapP$ such that
$x\in L$ if and only if $g(x)=0$.
By Lemma~\ref{lem:boolean-to-tn}, construct in polynomial time a rational-entry
scalar network $T_{g,x}$ on a connected square lattice with
\begin{equation}
    \cont(T_{g,x})=g(x),
    \qquad D(T_{g,x})=4.
\end{equation}
We modify the construction of Lemma~\ref{lem:saturated-gadget},
with $c=1/4$, by recording its branch label at two physical sites.
Choose sites $a$ and $b$ sites on opposite sides of the 
bipartition $A|B$.
As in Lemma~\ref{lem:saturated-gadget}, form the sitewise tensor
product $S_x:=T_{g,x}\otimes T_{g,x}$, whose contraction is
$g(x)^2$ and whose bond dimension is $16$.
Take three scalar branches on the same lattice: a copy of
$S_x$ with one tensor multiplied by $c$, a bond-dimension-one network of contraction one, and an unscaled copy of $S_x$. Apply the direct-sum construction of Lemma~\ref{lem:branch-amplitude}, assigning the same physical branch value at both marked sites $a$ and $b$: the resulting unnormalized PEPS is
\begin{equation}
    \ket*{\widetilde\Psi_x}
    =
    \left[
        c\,g(x)^2\ket{00}_{ab}
        +(1+g(x)^2)\ket{11}_{ab}
    \right]\otimes\ket{0}^{\otimes(N-2)}.
    \label{eq:nlm-unnormalized-gadget}
\end{equation}
Appending local pure ancillas leaves the nonzero Schmidt coefficients unchanged. Marking the second physical site does not enlarge the virtual
spaces, so the bond dimension is $16+1+16=33$, while the physical dimension is two.
Since $1+g(x)^2>0$, every output PEPS is nonzero.
Its normalized state is
\begin{equation}
    \ket{\Psi_x}
    =
    \frac{t\ket{00}_{ab}+\ket{11}_{ab}}{\sqrt{1+t^2}}
    \otimes\ket{0}^{\otimes(N-2)},
    \qquad
    t=F(g(x)),
    \label{eq:nlm-normalized-gadget}
\end{equation}
where $F$ is defined in Eq.~\eqref{eq:F-def}.
Here $t$ is the rational amplitude ratio; the Schmidt
coefficients are $t/\sqrt{1+t^2}$ and $1/\sqrt{1+t^2}$,
with zero coefficients omitted.
For a normalized state $u\ket{00}+v\ket{11}$,
Lemma~\ref{lem:nlm-zero} gives local-stabilizer equivalence precisely
when $uv=0$ or $|u|=|v|$.
In Eq.~\eqref{eq:nlm-normalized-gadget}, the second amplitude
never vanishes and $t\ge0$, so these cases reduce to $t=0$
or $t=1$. By Eq.~\eqref{eq:F-properties},
\begin{equation}
    g(x)=0\Longrightarrow t=0,
    \qquad
    g(x)\ne0\Longrightarrow \frac18\le t<\frac14.
\end{equation}
The equal-Schmidt-coefficient case is therefore excluded, and
\begin{equation}
    x\in L
    \quad\Longleftrightarrow\quad
    g(x)=0
    \quad\Longleftrightarrow\quad \ket{\Psi_x}\in\lstab
\end{equation}
This proves $\CeqP$-hardness of
Problem~\ref{prob:zero-nlm} under polynomial-time
many-one reductions.
\end{proof}

\section{Approximate PEPS stabilizer membership}
\label{sec:approx-membership}

We next consider a promise version of the stabilizer detection problem.

\begin{problem}[Approximate PEPS stabilizer membership]
\label{prob:approx-membership}
For efficiently computable thresholds $0\le a<b\le\sqrt2$, the promise problem $\ApproxSTABPEPS_{a,b}$ has as input a rational-entry tensor description $A$ of a nonzero $N$-qubit PEPS, with promise
\begin{equation}
\begin{aligned}
\yes: &\quad \dist(\ket{\psi(A)},\STAB_N)\le a,\\
\no:  &\quad \dist(\ket{\psi(A)},\STAB_N)\ge b.
\end{aligned}
\end{equation}
with $\ket{\psi(A)}$ the normalized state described by the PEPS input.
No condition is imposed on inputs whose distance lies in $(a,b)$.
\end{problem}

To quantify distance of a state from the stabilizer set, we use the Euclidean distance to pure stabilizer states,
\begin{equation}
        \dist(\ket*\psi,\STAB_N)
        :=
        \min_{\ket*\sigma\in\STAB_N}
        \min_{\theta\in\mathbb R}
        \norm{\ket*\psi-e^{i\theta}\ket*\sigma}.
\end{equation}
For normalized $\ket*\psi$ this is equivalent to the maximum stabilizer overlap:
\begin{equation}
        \dist(\ket*\psi,\STAB_N)^2
        =
        2-2\max_{\ket*\sigma\in\STAB_N}\abs{\braket*{\sigma}{\psi}}.
        \label{eq:distance-overlap}
\end{equation}

We now state that $\ApproxSTABPEPS_{a,b}$ is a $\CeqP-$hard problem even for a constant promise gap, and $\CeqP-$complete when $a=0$. 

\begin{theorem}[{Constant-threshold $\ApproxSTABPEPS_{a,b}$ is $\CeqP-$hard}]
\label{thm:constant-threshold}
Fix a rational constant $0<c<\sqrt2-1$ and define
\begin{equation}
\delta_c:=\left[2-2({1+(c/2)^2)^{-1/2}}\right]^{1/2}>0.
        \label{eq:delta-c-def}
\end{equation}
For every efficiently computable threshold pair $0\le a<b\le\delta_c$, the promise problem $\ApproxSTABPEPS_{a,b}$ is $\CeqP$-hard under polynomial-time promise many-one reductions.  In particular, choosing $c=1/4$ gives $d_1(1/8)>1/16$, and therefore $\ApproxSTABPEPS_{0,1/16}$ is $\CeqP$-complete as a promise problem.
\end{theorem}

Building on top of the previous hardness proofs, to prove hardness we again employ the encoding of the zero test of a $\GapP$ function in the single qubit amplitude of a suitable PEPS. More specifically, the steps of the proof are the following:
\begin{itemize}
    \item Encode the value of a $\GapP$ function $g$ in a scalar TN~(\cref{lem:boolean-to-tn}).
    \item Build a PEPS that realizes the state in~\cref{eq:omega-t-def}~(\cref{lem:branch-amplitude}).
    \item Prove that the distance of a state of the form $\ket\omega\otimes\ket0 ^{\otimes M}$ from the $M+1-$qubit stabilizer set is equal to the distance of $\ket\omega$ from the $1-$qubit stabilizer set~(\cref{lem:product-gadget-distance}).
    \item Find an analytical formula of the distance of $\omega_t$ as per~\cref{eq:omega-t-def} from $\STAB_1$~(\cref{lem:one-qubit-geometry}).
    \item Realize that a zero test of a $\GapP$ function automatically maps YES instances to stabilizer states, thus those with distance zero to $\STAB$ and NO instances to states with distance greater than $\delta_c$ thus establishing $\CeqP$ hardness.
\end{itemize}

We start by introducing the auxiliary lemmas previously mentioned, along with their proofs. The first one establishes an elementary fact of the stabilizer formalism, namely nonzero outcomes of partial computational-basis projection are proportional to stabilizer states.  

\begin{lemma}[Computational-basis projection of stabilizer states]
\label{lem:projection-stab}
Let $\ket*\Sigma\in\STAB_{r+s}$.  For every $z\in\{0,1\}^s$, the vector
\begin{equation}
        (\Id_r\otimes\bra z)\ket*\Sigma
\end{equation}
is either zero or proportional to an $r$-qubit stabilizer state.
\end{lemma}

\begin{proof}
Projecting a qubit onto $\ket*0$ or $\ket*1$ is postselection on a Pauli-$Z$ measurement outcome.  Stabilizer states are closed under Pauli measurements and postselection on any outcome of nonzero probability \cite{nielsen_chuang,Gottesman_1998}.  Iterating over the $s$ measured qubits proves the claim.
\end{proof}

The next lemma shows that the minimum distance between a single-qubit state $\ket\omega$ tensored with $M$ copies of the zero ket and the set of $(M+1)$-qubit stabilizer states is equal to the minimum distance between $\ket\omega$ and the single-qubit stabilizer set.

\begin{lemma}[Product stabilizer distance]
\label{lem:product-gadget-distance}
For every one-qubit state $\ket*\omega$ and every $M\ge0$,
\begin{equation}
        \dist(\ket*\omega\otimes\ket*0^{\otimes M},\STAB_{M+1})
        =
        \dist(\ket*\omega,\STAB_1).
\end{equation}
\end{lemma}
\begin{proof}
By \eqref{eq:distance-overlap}, it suffices to compare maximum stabilizer overlaps.  For any $\ket*\Sigma\in\STAB_{M+1}$,
\begin{equation}
        \braket*{\Sigma}{\omega\otimes0^M}
        =
        \left\langle
        (\Id\otimes\bra{0}^{\otimes M})\Sigma
        \middle|\omega
        \right\rangle .
\end{equation}
By \cref{lem:projection-stab}, the projected vector is either zero or has the form $\sqrt p\ket*\tau$ with $0<p\le1$ and $\ket*\tau\in\STAB_1$.  Hence the overlap is bounded by the maximum one-qubit stabilizer overlap with $\ket*\omega$.  The reverse inequality is achieved by the product stabilizer $\ket*\tau\otimes\ket*0^{\otimes M}$, where $\ket*\tau$ is a closest one-qubit stabilizer state to $\ket*\omega$.
\end{proof}

We will also need to derive the explicit formula of the minimum distance between a state of the form $\ket{\omega_t}\propto t\ket{0}+\ket{1}$ and the stabilizer set. The following lemma fills this gap.

\begin{lemma}[One-qubit $\ket{\omega_t}$ stabilizer distance]
\label{lem:one-qubit-geometry}
Let
\begin{equation}
        d_1(t):=\dist(\ket*{\omega_t},\STAB_1).
\end{equation}
with \begin{equation}
    \ket{\omega_t}=\frac{t\ket0+\ket1}{\sqrt{1+t^2}}
\end{equation}
For $|t|\le \sqrt2-1$,
\begin{equation}
        d_1(t)
        =
        \sqrt{2-\frac{2}{\sqrt{1+t^2}}}.
        \label{eq:d1-small}
\end{equation}
Moreover, $d_1(t)$ is strictly increasing as a function of $|t|$ on this interval.
\end{lemma}

\begin{proof}
The one-qubit stabilizer states are $\ket*0,\ket*1,\ket*\pm=(\ket*0\pm\ket*1)/\sqrt2$, and $\ket*{\pm i}=(\ket*0\pm i\ket*1)/\sqrt2$.  Direct comparison of overlaps gives the following:
\begin{equation}
\begin{array}{c|c}
\toprule
\text{candidate} & \abs{\braket*\sigma{\omega_t}} \\
\midrule
\ket*1 & (1+t^2)^{-1/2} \\
\ket*0 & |t|(1+t^2)^{-1/2} \\
\ket*+ & |1+t|[2(1+t^2)]^{-1/2} \\
\ket*- & |1-t|[2(1+t^2)]^{-1/2} \\
\ket*{\pm i} & 2^{-1/2} \\
\bottomrule
\end{array}
\end{equation}
For $|t|\le\sqrt2-1$, the largest overlap is with $\ket*1$, which gives \eqref{eq:d1-small}.  The monotonicity of $d_1$ follows immediately.  
\end{proof}

Using the preceding lemmas, the proof of~\cref{thm:constant-threshold} immediately follows.

\begin{proof}[Proof of \cref{thm:constant-threshold}]
Let $L\in\CeqP$, and choose $g\in\GapP$ such that $x\in L$ iff $g(x)=0$.  Construct $\ket{\Phi_g}$ using \cref{lem:saturated-gadget}.  If $g(x)=0$, then $F(g(x))=0$ and
\begin{equation}
        \ket{\Phi_g}=\ket1\otimes\ket0^{\otimes M}\in\STAB,
\end{equation}
so the distance to $\STAB$ is zero and the output is a YES instance for any $a\ge0$.

If $g(x)\ne0$, then $g(x)$ is a nonzero integer, so \eqref{eq:F-properties} gives $F(g(x))\ge c/2$.  Also $F(g(x))<c<\sqrt2-1$.  By \cref{lem:product-gadget-distance,lem:one-qubit-geometry},
\begin{equation}
        \dist(\ket{\Phi_g},\STAB)
        =
        d_1(F(g(x)))
        \ge
        d_1(c/2)
        =
        \delta_c.
\end{equation}
Thus nonzero GapP instances map to NO instances for every $b\le\delta_c$.  This is a polynomial-time promise many-one reduction from an arbitrary $\CeqP$ language.

For $a=0$, the exact membership test of \cref{thm:exact-membership} decides all promised instances: YES instances have distance zero and are exactly stabilizer states, while NO instances have positive distance.  Therefore $\ApproxSTABPEPS_{0,b}$ lies in the promise analogue of $\CeqP$.  With $c=1/4$ and $b=1/16$, this gives $\CeqP$-completeness.
\end{proof}
\subsection{Hardness of constant-accuracy SE evaluation}\label{sec:constant_acc_se}
The construction we employed in the proof of $\CeqP$ hardness of the stabilizer membership allows us also to prove $\CeqP$ hardness of approximate estimation of either stabilizer purity or entropy up to fixed constant accuracy.

\begin{theorem}[Hardness for constant-accuracy]
\label{cor:constant-accuracy-estimation}
Fix an integer $\alpha>1$ and define the positive rational constants
\begin{equation}
    \Delta_\alpha
    :=
    \frac12\left[
        1-\left(\frac{16}{65}\right)^{2\alpha}
         -\left(\frac{63}{65}\right)^{2\alpha}
    \right],
    \qquad
    \varepsilon_\alpha
    :=
    \frac{\Delta_\alpha}{4(\alpha-1)} \,.
    \label{eq:constant-estimation-accuracy}
\end{equation}
For nonzero rational-entry qubit PEPS on a square lattice with bond dimension $D=33$,
estimating either $\stp_\alpha$ or $M_\alpha$ of the normalized
state to additive error at most $\varepsilon_\alpha$ is
$\CeqP-$hard under polynomial-time Turing reductions.
More precisely, every language in $\CeqP$ can be decided using
one query to either estimation oracle and polynomial-time
postprocessing.
The same hardness holds for estimating $\stp_\alpha$ to relative
error at most $\varepsilon_\alpha$.
For $\alpha=2$, additive error $1/100$ suffices for both quantities,
as does relative error $1/100$ for stabilizer purity.
\end{theorem}

This theorem shows that a computational obstruction persists even when the requested numerical accuracy is fixed independently of the PEPS size.

\begin{proof}
Let $L\in\CeqP$, and choose $g\in\GapP$ such that
\begin{equation}
    x\in L\quad\iff\quad g(x)=0.
\end{equation}
Apply the construction of \cref{lem:saturated-gadget} with $c=1/4$. It produces, in polynomial time, a nonzero rational-entry square-lattice PEPS of constant bond dimension whose normalized state is
\begin{equation}
    \ket{\Phi_g}=\ket{\omega_{t}}\otimes\ket{0}^{\otimes M},\qquad\ket{\omega_t}=\frac{t\ket{0}+\ket{1}}{\sqrt{1+t^2}},\qquad t=\frac{g(x)^2}{4(1+g(x)^2)},
    \label{eq:constant-estimation-state}
\end{equation}
where $M$ is polynomially bounded in $|x|$. Since $g(x)$ is integer-valued,
\begin{equation}
    g(x)=0 \Longrightarrow t=0,\qquad g(x)\neq0 \Longrightarrow \frac18\le t<\frac14.
    \label{eq:constant-estimation-parameter-gap}
\end{equation}
A direct evaluation gives
\begin{equation}
    \stp_\alpha(\Phi_g)=p_\alpha(t),\qquad p_\alpha(t) := \frac12\left[1+\left(\frac{2t}{1+t^2}\right)^{2\alpha}+\left(\frac{1-t^2}{1+t^2}\right)^{2\alpha}\right].
    \label{eq:constant-estimation-one-qubit}
 \end{equation}
Recall that $p_\alpha$ is strictly decreasing on $0<t<\sqrt2-1$. Consequently,
\begin{equation}
    g(x)=0\Longrightarrow\stp_\alpha(\Phi_g)=1,\qquad g(x)\neq0\Longrightarrow\stp_\alpha(\Phi_g)\le p_\alpha(1/8)=1-\Delta_\alpha.
    \label{eq:constant-estimation-purity-gap}
\end{equation}
The constant $\Delta_\alpha$ is positive because $(16/65)^2+(63/65)^2=1$ and raising either of these two numbers to the power $\alpha>1$ strictly decreases it. Now suppose a stabilizer purity oracle returns $\widehat p$ satisfying
\begin{equation}
    \bigl|\widehat p-\stp_\alpha(\Phi_g)\bigr|\le\varepsilon_\alpha.
\end{equation}
Since, by choice, $\varepsilon_\alpha\le\Delta_\alpha/4$ we have
\begin{equation}
    x\in L \quad\iff\quad \widehat p>1-\frac{\Delta_\alpha}{2}.
\end{equation}
For entropy, the same construction gives
\begin{equation}
    g(x)=0\Longrightarrow M_\alpha(\Phi_g)=0,
\end{equation}
whereas, for $g(x)\neq0$,
\begin{align}
    M_\alpha(\Phi_g)
    &\ge
    -\frac{\log (1-\Delta_\alpha)}{\alpha-1}
    \notag\\
    &\ge
    \frac{\Delta_\alpha}{(\alpha-1)\ln2}
    \ge
    \frac{\Delta_\alpha}{\alpha-1}.
    \label{eq:constant-estimation-entropy-gap}
\end{align}
Here we used $-\ln(1-z)\ge z$ for $0<z<1$ and $\ln2<1$.
Thus an entropy oracle returning $\widehat M$ with
\begin{equation}
    \bigl|\widehat M-M_\alpha(\Phi_g)\bigr|
    \le\varepsilon_\alpha
\end{equation}
decides the same zero test by comparison
\begin{equation}
    x\in L
    \quad\iff\quad
    \widehat M<\frac{\Delta_\alpha}{2(\alpha-1)}
\end{equation}
All thresholds and error tolerances are fixed rational constants
for fixed $\alpha$. The construction and the final comparisons
therefore take polynomial time.

For relative purity error, an answer satisfying
\begin{equation}
    \bigl|\widehat p-\stp_\alpha(\Phi_g)\bigr|
    \le
    \varepsilon_\alpha\,\stp_\alpha(\Phi_g)
\end{equation}
also satisfies the additive-error condition, because
$\stp_\alpha(\Phi_g)\le1$. The same reduction therefore applies.
Finally, for $\alpha=2$,
\begin{equation}
    \Delta_2=\frac{1\,016\,064}{17\,850\,625},
    \qquad
    \varepsilon_2
    =\frac{254\,016}{17\,850\,625}
    >\frac1{100},
\end{equation}
which proves the stated numerical specialization.
\end{proof}

Note that \cref{cor:constant-accuracy-estimation} establishes
$\CeqP$-hardness of constant-accuracy numerical estimation.
Whether these estimation tasks are also $\sharpP$-hard at
constant additive accuracy remains open.

\end{document}